\documentclass[11pt]{article}
\usepackage[margin=1in]{geometry}
\usepackage[T1]{fontenc}
\usepackage[utf8]{inputenc}
\usepackage{lmodern,microtype,needspace}
\usepackage{amsmath,amssymb,amsthm,mathrsfs}
\usepackage{graphicx,booktabs,enumerate,nicefrac,soul,xcolor}
\usepackage{algorithm,algorithmic}
\usepackage[numbers,sort&compress]{natbib}
\usepackage{tikz,pgfplots}
\pgfplotsset{compat=1.18}
\usetikzlibrary{patterns,calc,arrows,shapes,automata,backgrounds,petri,snakes,matrix,positioning}
\usepackage[hidelinks]{hyperref}
\newtheorem{theorem}{Theorem}[section]
\newtheorem{lemma}[theorem]{Lemma}
\newtheorem{corollary}[theorem]{Corollary}
\newtheorem{definition}[theorem]{Definition}

\newcommand{\MMS}{\textsf{MMS}}
\newcommand{\WMMS}{\textsf{WMMS}}
\newcommand{\entitlement}{\textsf{w}}
\newcommand{\minp}{\textsf{minp}}
\newcommand{\kone}{\ell}
\newcommand{\ktwo}{r}
\newcommand{\best}{\mathcal{F}}
\newcommand{\cost}{V}
\newcommand{\chores}{M}
\newcommand{\agents}{N}
\newcommand{\alloc}{A}
\newcommand{\agent}{a}

\newcommand{\chore}{b}

\newcommand{\vthree}{2.06447}
\newcommand{\vfour}{2.47471}
\newcommand{\vfive}{2.73187}
\newcommand{\vsix}{3.04882}
\newcommand{\vseven}{3.2842}
\newcommand{\veight}{3.5134}
\newcommand{\vnine}{3.72934}
\newcommand{\vten}{4.0352}

\newlength{\chainlhswidth}
\newsavebox{\equationannotationbox}
\newcommand{\annotatedrow}[2]{%
\sbox{\equationannotationbox}{$\displaystyle #1\quad #2$}%
\ifdim\wd\equationannotationbox<\dimexpr\columnwidth-24pt-\chainlhswidth\relax
\def\nextrow{#1\quad #2}%
\else
\def\nextrow{\begin{aligned}[t]&#1\\&\qquad #2\end{aligned}}%
\fi\nextrow}

\newcommand{\restatedname}{}

\title{Improved Weighted Maximin Share Guarantees for Indivisible Chores}
\author{Masoud Seddighin \and Saeed Seddighin}
\date{}
\hypersetup{pdftitle={Improved Weighted Maximin Share Guarantees for Indivisible Chores},pdfauthor={Masoud Seddighin and Saeed Seddighin}}
\begin{document}
\maketitle
\begin{abstract}
We consider the problem of assigning indivisible chores to agents with different entitlements in the maximin share (\MMS) context. While constant-\MMS\ allocations are guaranteed to exist for both goods and chores in the symmetric setting, the situation becomes much more complex when agents have different entitlements. For the allocation of indivisible goods, it has been proven that an $n$-\WMMS\ (weighted \MMS) guarantee is the best one can hope for. For indivisible chores, however, it was recently proved that an $O(\log n)$-\WMMS\ assignment is guaranteed to exist. In this work, we improve this upper bound to a constant-\WMMS\ guarantee.
We prove the existence of a $\tfrac{54}{5}$-$\WMMS$ assignment using a combinatorial layered moving-knife algorithm.\footnote{For concurrent and independent work by Li, Wang, and Xing~\cite{li2025constant}, see the \hyperref[remark:parallel-work]{remark on parallel work} in the introduction.} For every fixed $\epsilon>0$, this yields a $(\tfrac{54}{5}+\epsilon)$-$\WMMS$ assignment in polynomial time. Using a faster estimate of the weighted maximin shares instead gives a $\tfrac{108}{5}\cdot (1+\epsilon)$-$\WMMS$ assignment in almost linear time.

We also develop a chore-oblivious analysis that guarantees\linebreak $\WMMS$ approximation factors of approximately $2.06447$, $2.475855$, and $2.732051$ for three, four, and five agents, respectively.

\end{abstract}
\section{Introduction}
In this work, we study the problem of fair assignment of indivisible chores to $n$ agents who have \textbf{unequal entitlements}. Fair division is a fundamental problem that has been extensively studied in both computer science and economics~\cite{Procaccia:first,amanatidis2015approximation,caragiannis2016unreasonable,ghodsi2018fair,aziz2016approximation}. The origins of this problem date back to 1948, when Steinhaus introduced the classic cake-cutting problem: given $n$ agents with distinct valuation functions over a cake, can we divide the cake so that each agent receives a share whose value is at least $1/n$ of the total value? Dubins and Spanier provided a positive answer using the elegant moving knife method \cite{Dubins:first}. Although the basic problem has a straightforward solution, many variants have emerged over the years, several of which remain open despite decades of research~\cite{brams1996fair,robertson1998cake,Procaccia:first,aziz2022algorithmic,procaccia2013cake,guo2023survey}. Recently, attention has shifted to the fair allocation/assignment \footnote{For chores, the literature often uses the term allocation. Here we use assignment instead, since allocation typically refers to distributing valuable items.} of indivisible goods/chores, which differs from the continuous nature of cake-cutting. Instead of dividing a divisible resource, we consider a set of indivisible items to be allocated/assigned among $n$ agents. A proportional division is often infeasible due to obvious counterexamples such as when there is only one item to assign.

The dual problem to fair allocation of goods is called fair assignment of chores, which is the focus of this work. In this variant, instead of goods that are desirable for the agents, we have chores that incur costs to the agents. We aim to divide a set of chores among a number of agents in a fair manner. That is, no agent is assigned more than a fair load of chores. Our work focuses  on the indivisible chores setting and uses maximin-share value (\MMS) as the notion of fairness proposed by Budish~\cite{Budish:first}.

Let $\agents=\{\agent_1, \agent_2,\allowbreak \ldots, \agent_n\}$ be a set of $n$ agents, and $\chores = \{\chore_1, \chore_2,\allowbreak\ldots,\chore_m\}$ be a set of $m$ chores. Each agent $\agent_i$ has an \textit{additive} function $\cost_i$ over the chores which denotes her cost for each subset of chores. We assume that $\cost_i(\{\chore_j\})>0$ for every agent $\agent_i$ and every chore $\chore_j$. For the symmetric setting (agents with equal entitlements), let $\Pi_n(\chores)$ denote the set of all ordered partitions of the chores into $n$ bundles, allowing empty bundles. The maximin-share  value of agent $\agent_i$ (denoted by $\MMS_i$) is defined as
\begin{equation*}
	\MMS_i = \min_{\langle A_1, A_2,\allowbreak \ldots, A_n\rangle \in \Pi_n(\chores)} \max_{\agent_j \in \agents} \cost_i(A_j).
\end{equation*}

In this setting, our aim is to find an assignment of chores to the agents such that the cost of each agent for the chores assigned to her is bounded by a multiplicative factor of her \MMS\ value. The \MMS\ value  serves as a natural benchmark—since, even when all agents have identical valuations, one agent must be assigned a bundle at least as costly as her \MMS\ value. Thus, the question that arises is whether there are assignments for which the costs of the chores given to agents are bounded by a constant factor of their \MMS\ value. The best known upper bound is $\tfrac{13}{11}$~\cite{huang2023reduction}.

While the case of equal entitlements is widely studied for both fair allocation and chore assignment problems, cases of unequal entitlements arise frequently in real-world scenarios. For example, inheritance laws in various cultures and religions often mandate unequal divisions. Similarly, assets of a bankrupt company must be divided among its stakeholders according to their claims, and the division of natural resources is typically based on geographic or political considerations. 

While constant-approximation guarantees are straightforward for the symmetric settings~\cite{aziz2017algorithms}, the problem becomes much more challenging when agents have different entitlements.
 In this version, each agent $\agent_i$ has an entitlement $\entitlement_i$ such that $\sum \entitlement_i = 1$. Unfortunately, the previous definition of \MMS\  does not take into account the entitlements of the agents. To address this, the following interpretation of \MMS\ values for the asymmetric setting known as \WMMS\ values is proposed by Aziz, Chan, and Li~\cite{aziz2019weighted} (this definition is based on a similar definition for the case of fair allocation of indivisible goods proposed by Farhadi \textit{et al.}~\cite{farhadi2017fair}).
 \begin{equation*}
 	\WMMS_i = \min_{\langle A_1, A_2,\allowbreak \ldots, A_n\rangle \in \Pi_n(\chores)} \max_{\agent_j \in \agents} \cost_i(A_j)\frac{\entitlement_i}{\entitlement_j}.
\end{equation*}
We define an assignment as being \emph{$\alpha$-\WMMS} if all chores are assigned to agents and each agent $\agent_i$ receives a bundle $S_i$ of chores such that $\cost_i(S_i) \leq \alpha \cdot \WMMS_i$. While for both goods and chores, constant-\MMS\ allocations/assignments are guaranteed to exist for the symmetric setting, the situation is much more complex when agents have different entitlements. For  goods, it has been proven that $n$-\WMMS\ is the best guarantee one can hope for~\cite{farhadi2017fair}. 

Aziz \textit{et al.}~\cite{aziz2019weighted} initiated the study of fair assignment of indivisible chores with unequal entitlements. 
 They show that commonly used algorithms that work well for the allocation of goods to asymmetric agents, and even for chores to symmetric agents, do not provide any desirable guarantee for the asymmetric chore division problem. In addition, they establish a lower bound of $\tfrac{4}{3}$ and an upper bound of $\tfrac{3}{2}$ on the optimal worst-case $\WMMS$ approximation factor for two agents. Subsequent work by Wang, Li, and Lu~\cite{wang2024improved} showed that an $O(\log n)$-\WMMS\ assignment exists for any number of agents. In addition, they proved that the optimal assignment for the case of two agents is $\tfrac{\sqrt{3}+1}{2}$-\WMMS.

Our main contribution is a proof of the existence of a constant-\WMMS\ assignment for any number of agents, based on a novel \textit{layered moving knife} algorithm that extends the moving knife to the asymmetric setting.
\begin{table}[t]
\centering
\small
\begin{minipage}[t]{0.43\textwidth}
\centering
\textbf{(A)}\par\smallskip
\setlength{\tabcolsep}{4pt}
\begin{tabular*}{\linewidth}{@{\extracolsep{\fill}}ccc@{}}
\toprule
$n$ & Prior work~\cite{wang2024improved} & \shortstack{Our chore-oblivious\\bounds} \\
\midrule
2 & $\tfrac{\sqrt{3}+1}{2}$ & -- \\
3 & $\approx 2.585$ & $\approx2.06447$ \\
4 & $\approx 3$ & $\approx2.475855$ \\
5 & $\approx 3.32$ & $\approx2.732051$ \\
\bottomrule
\end{tabular*}
\end{minipage}\hfill
\begin{minipage}[t]{0.45\textwidth}
\centering
\textbf{(B)}\par\smallskip
\setlength{\tabcolsep}{3pt}
\begin{tabular*}{\linewidth}{@{\extracolsep{\fill}}lcc@{}}
\toprule
& Existential & Polytime \\
\midrule
Prior work~\cite{wang2024improved} & $\log_2 n+1$ & -- \\
\shortstack[l]{Concurrent work\\\cite{li2025constant}} & $12$ & $24+\epsilon$ \\
Our results & \shortstack{$\tfrac{54}{5}=10.8$\\Theorem~\ref{theorem:main}} & \shortstack{$10.8+\epsilon$\\Corollary~\ref{cor:ptas-implementation}} \\
\shortstack[l]{Our almost-linear\\implementation} & -- & \shortstack{$21.6\cdot(1+\epsilon)$\\Theorem~\ref{theorem:main_computation}} \\
\bottomrule
\end{tabular*}
\end{minipage}
\caption{Approximation guarantees: (A) for a small number of agents; (B) for an arbitrary number of agents.}
\label{tab:results}
\end{table}

\vspace{0.2cm}
{\noindent \textbf{Theorem} \ref{theorem:main} [restated]. \textit{The asymmetric chore division problem admits a $\tfrac{54}{5}$-\WMMS\ assignment.\\}}

To prove Theorem \ref{theorem:main}, we first show that we can assume without loss of generality that if each agent sorts the chores based on her cost function, the ordering would be the same for all agents. In a nutshell, our layered moving knife algorithm divides the agents into three categories: 
 dead agents, agents in progress, and agents in the queue.
We sort the chores in a particular order (more details can be found in Section~\ref{sec:constant}) and iterate over the chores one by one. We also construct a multiset that only contains agents in progress but each such agent may appear several times in that multiset. Each time we add the corresponding chore to a bag as long as there is one agent in our multiset whose cost for the bag is bounded by a certain value. Once no other chore can be added to the bag, we give all chores in the bag to one agent of our multiset and remove that agent from the multiset. Once the multiset is empty, we move all agents in progress to the set of dead agents and bring some of the queued agents into the set of agents in progress. We prove in Section~\ref{sec:constant} that layered moving knife gives a $\tfrac{54}{5}$-\WMMS\ assignment.

Our algorithm is simple and admits an implementation in almost linear time. The only challenge is that the $\WMMS$ values must be known in advance. In Section~\ref{sec:comp}, we show how to compute a $2(1+\epsilon)$-approximation of  $\WMMS_i$ in almost linear time. Combining this approximation procedure with the layered moving-knife algorithm yields a $\tfrac{108}{5}\cdot (1+\epsilon)$-$\WMMS$ assignment in almost linear time. Alternatively, using the PTAS for approximating each agent's optimal $\WMMS$ partition~\cite{hochbaum1988uniform} yields a $(\tfrac{54}{5}+\epsilon)$-$\WMMS$ assignment in polynomial time for every fixed $\epsilon>0$ (Corollary~\ref{cor:ptas-implementation}).


Although a constant approximation improves on a logarithmic guarantee in general, our factor $\tfrac{54}{5}$ exceeds $\log_2 n+1$ for small numbers of agents. Thus, it is natural to study these cases separately and seek better approximation guarantees. For $n=2$, Wang, Li, and Lu~\cite{wang2024improved} proved that the optimal worst-case approximation factor is $\tfrac{\sqrt{3}+1}{2}$.
We also make progress on this front. We devise a \textit{chore-oblivious analysis} for the asymmetric chore-division problem that improves the  $\log_2 n+1$ upper bound for $n=3,4,5$. We also present empirical results for $3\leq n\leq10$: for each $n$, the largest upper bound obtained by our chore-oblivious technique over the sampled entitlement vectors is below $\log_2 n+1$.

 Let us define the notation $\best(\langle \entitlement_1, \entitlement_2,\allowbreak \ldots, \entitlement_n\rangle)$, which denotes the smallest approximation factor guaranteed to exist for every instance with entitlements $\entitlement_1, \entitlement_2,\allowbreak \ldots, \entitlement_n$. We use the following known bounds:
\begin{center}
	\renewcommand{\arraystretch}{1.4}
	\begin{tabular}{lll}
		\(\best(\langle 1 \rangle)\) & \(= 1\) \\
		\(\best(\langle \tfrac{1}{2}, \tfrac{1}{2} \rangle)\) & \(= 1\) \\
		\(\best(\langle w, 1-w \rangle)\) & \(\leq \tfrac{\sqrt{3}+1}{2}\)&~\cite{wang2024improved} \\
		\(\best(\langle \tfrac{1}{3}, \tfrac{1}{3}, \tfrac{1}{3} \rangle)\) & \(\leq \tfrac{19}{18}\)&~\cite{feige2022improved} \\
		\(\best(\langle 1/n,1/n,\allowbreak \ldots, 1/n \rangle)\) & \(\leq \tfrac{20}{17}\), for \(4 \leq n \leq 7\)&~\cite{huang2023reduction} \\
		\(\best(\langle 1/n,1/n,\allowbreak \ldots, 1/n \rangle)\) & \(\leq \tfrac{13}{11}\), for \(n > 7\)&~\cite{huang2023reduction}
	\end{tabular}
\end{center}

We next present two simple reductions that derive new bounds beyond those listed above.

\vspace{0.2cm}
{\noindent \textbf{Lemma} \ref{red1} [restated]. \textit{	Let \(\entitlement_1, \entitlement_2,\allowbreak \ldots, \entitlement_n\) denote the entitlements of \(n\) agents. Suppose the agents are partitioned into \(n'\) disjoint groups, and for each group \(i \in [n']\), let \(i^* \in [n]\) be the index of a designated representative agent (not necessarily in group \(i\)), with representatives pairwise distinct. Let \(\alpha > 1\), and assume the total entitlement in group \(i\) is at most \(\alpha \entitlement_{i^*}\), i.e.,
$
			\sum_{j \in \text{group } i} \entitlement_j \leq \alpha \entitlement_{i^*}.
$		
		Then, for $\beta = \tfrac{1}{\left( \sum_{i=1}^{n'} \entitlement_{i^*} \right)}$, we have
		$
		\best\left(\langle\entitlement_1, \entitlement_2,\allowbreak \ldots, \entitlement_n\rangle\right) \leq \alpha \cdot \best(\langle\beta \entitlement_{1^*}, \beta \entitlement_{2^*},\allowbreak \ldots, \beta \entitlement_{{n'}^*}\rangle).
		$
}}

{\noindent \textbf{Lemma} \ref{red2} [restated]. \textit{	Let \(\entitlement_1, \entitlement_2,\allowbreak \ldots, \entitlement_n\) be the entitlements of \(n\) agents. Then:
\[
		\best(\langle\entitlement_1, \entitlement_2,\allowbreak \ldots, \entitlement_n\rangle) \leq \frac{\max_i \entitlement_i}{\min_i \entitlement_i} \cdot \best(\langle \tfrac{1}{n}, \tfrac{1}{n},\allowbreak \ldots, \tfrac{1}{n} \rangle).
		\]}}
	
Using Lemmas~\ref{red1} and~\ref{red2}, we improve the existing bounds for $n=3,4,5$. These reductions yield better upper bounds for small $n$. We show that for $n=3,4,5$,  $\best(\langle \entitlement_1, \entitlement_2,\allowbreak \ldots, \entitlement_n \rangle) $ is bounded by certain upper bounds for any sequence of entitlements $\entitlement_1, \entitlement_2,\allowbreak \ldots, \entitlement_n$ that sum to $1$. In addition to this, for each $3\leq n\leq10$, we generate a billion entitlement vectors by normalizing independent uniform random draws to sum to $1$. We use our chore-oblivious technique to compute an upper bound on $\best$ for each sampled vector and report the largest bound observed. For $n=3,4,5$, the largest observed values are close to our analytical upper bounds on $\best$. This numerical comparison alone does not establish tightness of the chore-oblivious method. For $6\leq n\leq10$, the computed bounds improve on $\log_2 n+1$ for all sampled entitlement vectors, providing empirical evidence for a possible improvement of the general guarantee.

\paragraph{A note on parallel work.}
\label{remark:parallel-work}
Before concluding this section, we note that concurrently with and independently of our work, a manuscript by Li, Wang, and Xing~\cite{li2025constant} was posted on arXiv, which also establishes the existence of a constant-$\WMMS$ guarantee for the same problem. The two works were developed independently.
Our combinatorial approach improves both the existential and the polynomial-time approximation guarantees. Li, Wang, and Xing~\cite{li2025constant} establish an existential factor of $12$, while their polynomial-time implementation uses linear programming and rounding to obtain $24+\epsilon$. Our layered moving-knife algorithm gives a factor of $\tfrac{54}{5}=10.8$ and, using the PTAS for approximating each agent's optimal $\WMMS$ partition~\cite{hochbaum1988uniform}, achieves $\tfrac{54}{5}+\epsilon$ in polynomial time. Moreover, a faster WMMS estimation procedure yields an almost-linear-time implementation with factor $\tfrac{108}{5}\cdot(1+\epsilon)=21.6(1+\epsilon)$.

\section{Related Work}
Fair allocation/assignment of indivisible items encompasses a broad range of settings, fairness notions, and practical scenarios. Several surveys provide comprehensive overviews of this area. For a focused review of the allocation of indivisible chores, see ~\cite{guo2023survey}. Broader discussions that cover both goods and chores can be found in ~\cite{aziz2022fair,aziz2022algorithmic} and ~\cite{amanatidis2022fair}. For classical perspectives on fair division and cake-cutting problems, we refer to ~\cite{procaccia2013cake} and ~\cite{brams1996fair}. For a review of fair division with unequal entitlements, see ~\cite{suksompong2025weighted}. Here, we focus  on chores and the maximin share notion.

For additive costs and equal entitlements, Huang and Lu~\cite{huang2021framework} developed an algorithmic framework giving an $\tfrac{11}{9}$-$\MMS$ guarantee and explored its connection to job scheduling. Aziz, Li, and Wu~\cite{aziz2024ordinal} obtained a deterministic $\tfrac{5}{3}$-$\MMS$ guarantee using only ordinal preferences and also studied strategyproof approximation algorithms. Beyond additive costs, Li, Wang, and Zhou~\cite{li2023beyond} established impossibility results for submodular costs and approximation guarantees for subadditive costs, including constant guarantees for bin-packing and job-scheduling costs.

Many of the main results on maximin share assignments of indivisible chores—particularly in the asymmetric setting—have been discussed in the introduction. Here, we highlight additional results related to the maximin share notion for chores. Several results establish \MMS\ guarantees under specific assumptions on the number of agents, the number of chores, or the structure of cost functions. Hummel~\cite{hummel2023lower} proved that for any constant \( c \), there exists a large enough \( n \) such that every instance with \( n \) agents and \( n + c \) chores admits an $\MMS$ allocation.
Barman \textit{et al.}~\cite{barman2023fair} showed that \MMS\ allocations are guaranteed when agents have binary supermodular cost functions. Feige \textit{et al.}~\cite{feige2021tight} constructed an instance with three agents and nine chores in which every assignment gives at least one agent a bundle with cost at least \( \tfrac{44}{43} \) of their \MMS\ value. In the context of ordinal approximation, Hosseini \textit{et al.}~\cite{hosseini2022ordinal} showed that a relaxation of $\MMS$ called \( 1 \)-out-of-\( \lfloor \tfrac{3}{4}\cdot n \rfloor \) \MMS\ allocation always exists. Garg, Huang, and Segal-Halevi~\cite{garg2025binpacking} improved this ordinal guarantee to $1$-out-of-$\lfloor 9n/11\rfloor$ \MMS\ for additive costs.

For chores with unequal entitlements, Wu, Zhang, and Zhou~\cite{wu2025weighted} studied weighted envy-freeness up to one item (WEF1). They gave an  algorithm for computing WEF1 assignments under additive costs and efficient algorithms for assignments that satisfy both WEF1 and Pareto optimality for bivalued instances.

Related to our study is the concept of Anyprice Share, introduced by Babaioff, Ezra, and Feige ~\cite{babaioff2024fair}, which applies to fair allocation settings with arbitrary entitlements. Recently, Babaioff and Feige ~\cite{babaioff2025share} proposed a general framework for share-based fairness under arbitrary entitlements, encompassing many share-based notions including  maximin share. They designed randomized algorithms that offer both ex-ante and ex-post guarantees relative to any specified share. For indivisible goods with arbitrary entitlements, Feige and Grinberg~\cite{feige2025fair} established ex-post Anyprice Share guarantees and ex-ante maximum expectation share guarantees under subadditive and XOS valuations.

\section{Preliminaries}
For a positive integer $t$, write $[t]=\{1,\allowbreak\ldots,t\}$.
Let $\agents$ be a set of $n>1$ agents, and $\chores$ be a set of $m>1$ chores. We denote the agents by $\agents=\{\agent_1, \agent_2,\allowbreak \ldots, \agent_n\}$ and chores by $\chores = \{\chore_1, \chore_2,\allowbreak\ldots,\chore_m\}$. Each agent $\agent_i$ has an \textit{additive} function $\cost_i$ over the chores which denotes her cost for each subset of chores. We assume that $\cost_i(\{\chore_j\})>0$ for every agent $\agent_i$ and every chore $\chore_j$. Additionally, each agent $\agent_i$ has an entitlement $0 < \entitlement_i \leq 1$ such that the entitlements sum to $1$. We assume without loss of generality that $\entitlement_1 \leq \entitlement_2 \leq \ldots \leq \entitlement_n$.
An assignment is an ordered partition $\langle S_1,\allowbreak\ldots,S_n\rangle$ of $\chores$, where $S_i$ is assigned to agent $\agent_i$, $\bigcup_{i=1}^{n}S_i=\chores$, and $S_i\cap S_j=\emptyset$ for all distinct $i,j\in[n]$. Bundles may be empty.

Since our model generalizes the notion of the maximin share, we begin by recalling its definition for equal entitlements, as proposed by Budish~\cite{Budish:first}. In this case, we assume that all entitlements are equal to $1/n$. Let $\Pi_n(\chores)$ denote the set of all ordered partitions of the chores into $n$ bundles, allowing empty bundles. The maximin  share value of agent $\agent_i$ (denoted by $\MMS_i$) is defined as
\begin{equation}\label{eq:1}
	\MMS_i = \min_{\langle A_1, A_2,\allowbreak \ldots, A_n\rangle \in \Pi_n(\chores)} \max_{\agent_j \in \agents} \cost_i(A_j).
\end{equation}

This measure can be interpreted as the outcome of a cautious agent in a divide-and-choose procedure where the agent only knows their own cost function. The agent divides the chores into $n$ bundles, aiming to make the most costly bundle (to themselves) as desirable as possible. However, when agents have unequal entitlements, this interpretation fails, as the divide-and-choose approach must now account for differing entitlements. To address this, the following definition is proposed by Aziz, Chan, and Li~\cite{aziz2019weighted}:
\begin{equation}\label{eq:2}
	\WMMS_i = \min_{\langle A_1, A_2,\allowbreak \ldots, A_n\rangle \in \Pi_n(\chores)} \max_{\agent_j \in \agents} \cost_i(A_j)\frac{\entitlement_i}{\entitlement_j}.
\end{equation}
Notice that Equations \eqref{eq:1} and~\eqref{eq:2} are equivalent when $\entitlement_i=1/n$ for all agents.
We define an assignment as  \emph{$\alpha$-\WMMS} if all chores are assigned to agents and each agent $\agent_i$ receives a bundle $S_i$ of chores such that $\cost_i(S_i) \leq \alpha \cdot \WMMS_i$.

\section{Layered Moving Knife Algorithm}\label{sec:constant}
We show in this section that there always exists an assignment of the chores to the agents such that each agent receives a bundle that costs at most 54/5 times her weighted maximin share value.

 We begin this section by stating a standard reduction to the \textit{sorted-chores} setting. In this setting, we have additional constraints ensuring  $\cost_i(\{\chore_1\}) \leq \cost_i(\{\chore_2\}) \leq \ldots \leq \cost_i(\{\chore_m\})$ for each agent $\agent_i \in \agents$. More precisely, the existence of an $\alpha$-\WMMS\ assignment in the sorted-chores setting implies its existence in the general setting. We then state our \textit{layered moving knife} algorithm that obtains an assignment which is constant-\WMMS.

\begin{lemma}[\cite{barman2017approximation}, restated ]\label{lemma:1}
If an $\alpha$-\WMMS\ assignment is guaranteed to exist in the sorted-chores setting, then such an $\alpha$-\WMMS\ assignment also exists for the general asymmetric chore division problem. Making this assumption adds $O(nm \log m)$ time to the running time of the assignment algorithm.

\end{lemma}
\begin{proof}
	Let $I$ be an instance of the chore division problem. We first construct an instance $I'$ of the chore division problem that is identical to $I$ except that the cost functions are sorted. Denote the cost functions of $I'$ by $\cost'_1, \cost'_2,\allowbreak \ldots, \cost'_n$, where each $\cost'_i$ is defined as follows: for each $j \in [m]$, let $\cost'_i(\{\chore_j\})$ be the $j$-th smallest value in the multiset 
	\[
	\{\cost_i(\{\chore_1\}),\allowbreak \ldots, \cost_i(\{\chore_m\})\}.
	\]
	
	By construction, $I'$ satisfies the criteria of the sorted-chores setting, and thus an $\alpha$-\WMMS\ assignment exists. Fix one such assignment as the solution for $I'$. 
	
	We now construct an $\alpha$-\WMMS\ assignment for the original instance $I$ as follows. Iterate over the chores in $I'$ from $\chore_1$ to $\chore_m$, and for each chore $\chore_j$, let $\agent_i$ be the agent who receives $\chore_j$ in the solution of $I'$. In instance $I$, assign to $\agent_i$ one of the $j$ chores that has the smallest cost to her and has not yet been assigned. Since fewer than $j$ chores have been allocated so far, such a chore always exists. By definition, the cost of $\agent_i$ for this chore is at most $\cost'_i(\{\chore_j\})$.
	
	Proceeding in this manner, we construct an assignment for instance $I$ in which the total cost of each agent for all assigned chores is bounded by the same quantity as in $I'$. Moreover, the $\WMMS_i$ values remain the same for all agents in both instances. Therefore, the resulting assignment is also $\alpha$-\WMMS\ for instance $I$.
	
	\paragraph{Time complexity.} Constructing the sorted cost functions $\cost'_i$ for all agents requires sorting $m$ chores for each of the $n$ agents, which takes
	$
	O(n m \log m) \text{ time.}
	$
	After finding a solution to $I'$, we convert it to a solution for $I$ as follows: scanning the chores in $I'$ in non-decreasing order of their cost to every agent ( $\chore_1,\allowbreak\ldots,\chore_m$), if a chore $\chore_j$ is assigned to agent $\agent_i$, then we find the smallest chore in the sorted list of chores of agent $\agent_i$ in $I$ which is not yet assigned, and assign it to agent $\agent_i$. 
	
	Assigning the chores in $I$ according to the solution of $I'$ requires scanning the chores in $I'$ in  order of  $\chore_1,\allowbreak\ldots,\chore_m$ and, for each, selecting the cheapest unassigned chore for the corresponding agent. This can be done in $O(n m)$ time by maintaining sorted lists of chores for each agent and a pointer to the index of the last assigned chore to that agent. To find the smallest unassigned chore, we increase the pointer corresponding to that agent until we find an unassigned chore in the list. We also update the pointer of that agent. This way, we need to scan the sorted list of each agent once overall, which for $n$ agents takes $O(nm)$ time. Therefore, the total running time of this procedure is
	\[
	O(n m \log m + n m) = O(n m \log m).
	\]
\end{proof}

Based on Lemma~\ref{lemma:1}, we assume without loss of generality that we meet the conditions of the sorted-chores setting. Also, we assume from here on that $\WMMS_i=\entitlement_i$ for all $\agent_i\in\agents$. Since scaling the cost functions does not affect the approximation guarantees, this comes without loss of generality.

\begin{figure}[t]
	\begin{center}
		\tikzset{every picture/.style={line width=0.35pt}} 
		\scalebox{0.65}{
			\begin{tikzpicture}[x=0.75pt,y=0.75pt,yscale=-0.9,xscale=0.9]
				
				\draw   (118,109.5) -- (156.5,109.5) -- (156.5,148) -- (118,148) -- cycle ;
				\draw   (168,109.5) -- (206.5,109.5) -- (206.5,148) -- (168,148) -- cycle ;
				\draw   (218,109.5) -- (256.5,109.5) -- (256.5,148) -- (218,148) -- cycle ;
				\draw   (268,109.5) -- (306.5,109.5) -- (306.5,148) -- (268,148) -- cycle ;
				\draw   (368,109.5) -- (406.5,109.5) -- (406.5,148) -- (368,148) -- cycle ;
				\draw   (418,109.5) -- (456.5,109.5) -- (456.5,148) -- (418,148) -- cycle ;
				\draw   (468,109.5) -- (506.5,109.5) -- (506.5,148) -- (468,148) -- cycle ;
				\draw   (518,109.5) -- (556.5,109.5) -- (556.5,148) -- (518,148) -- cycle ;
				\draw    (536.5,82) -- (536.5,107) ;
				\draw [shift={(536.5,110)}, rotate = 270] [fill={rgb, 255:red, 0; green, 0; blue, 0 }  ][line width=0.08]  [draw opacity=0] (10.72,-5.15) -- (0,0) -- (10.72,5.15) -- (7.12,0) -- cycle    ;
				\draw   (118,199.5) -- (156.5,199.5) -- (156.5,238) -- (118,238) -- cycle ;
				\draw   (168,199.5) -- (206.5,199.5) -- (206.5,238) -- (168,238) -- cycle ;
				\draw   (218,199.5) -- (256.5,199.5) -- (256.5,238) -- (218,238) -- cycle ;
				\draw   (268,199.5) -- (306.5,199.5) -- (306.5,238) -- (268,238) -- cycle ;
				\draw   (368,199.5) -- (406.5,199.5) -- (406.5,238) -- (368,238) -- cycle ;
				\draw   (418,199.5) -- (456.5,199.5) -- (456.5,238) -- (418,238) -- cycle ;
				\draw   (468,199.5) -- (506.5,199.5) -- (506.5,238) -- (468,238) -- cycle ;
				\draw   (518,199.5) -- (556.5,199.5) -- (556.5,238) -- (518,238) -- cycle ;
				\draw  [fill={rgb, 255:red, 248; green, 231; blue, 28 }  ,fill opacity=0.32 ] (105.5,184) -- (512.5,184) -- (512.5,250) -- (105.5,250) -- cycle ;
				\draw  [fill={rgb, 255:red, 80; green, 227; blue, 194 }  ,fill opacity=0.34 ] (512.5,184) -- (565.5,184) -- (565.5,250) -- (512.5,250) -- cycle ;
				
				\draw (130,119.4) node [anchor=north west][inner sep=0.75pt]    {$b_{1}$};
				\draw (180,119.4) node [anchor=north west][inner sep=0.75pt]    {$b_{2}$};
				\draw (228,119.4) node [anchor=north west][inner sep=0.75pt]    {$b_{3}$};
				\draw (278,119.4) node [anchor=north west][inner sep=0.75pt]    {$b_{4}$};
				\draw (318,125.4) node [anchor=north west][inner sep=0.75pt]  [font=\Huge]  {$\dotsc $};
				\draw (372,119.4) node [anchor=north west][inner sep=0.75pt]    {$b_{m-3}$};
				\draw (421,118.4) node [anchor=north west][inner sep=0.75pt]    {$b_{m-2}$};
				\draw (471,118.4) node [anchor=north west][inner sep=0.75pt]    {$b_{m-1}$};
				\draw (528,118.4) node [anchor=north west][inner sep=0.75pt]    {$b_{m}$};
				\draw (532,62.4) node [anchor=north west][inner sep=0.75pt]    {$s$};
				\draw (130,215.4) node [anchor=north west][inner sep=0.75pt]    {$a_{1}$};
				\draw (180,215.4) node [anchor=north west][inner sep=0.75pt]    {$a_{2}$};
				\draw (228,215.4) node [anchor=north west][inner sep=0.75pt]    {$a_{3}$};
				\draw (278,215.4) node [anchor=north west][inner sep=0.75pt]    {$a_{4}$};
				\draw (318,225.4) node [anchor=north west][inner sep=0.75pt]  [font=\Huge]  {$\dotsc $};
				\draw (372,215.4) node [anchor=north west][inner sep=0.75pt]    {$a_{n-3}$};
				\draw (421,215.4) node [anchor=north west][inner sep=0.75pt]    {$a_{n-2}$};
				\draw (471,215.4) node [anchor=north west][inner sep=0.75pt]    {$a_{n-1}$};
				\draw (528,215.4) node [anchor=north west][inner sep=0.75pt]    {$a_{n}$};
				\draw (104,167) node [anchor=north west][inner sep=0.75pt]   [align=left] {$Q$};
				\draw (514,167) node [anchor=north west][inner sep=0.75pt]   [align=left] {$P$};
			\end{tikzpicture}
		}
	\end{center}
	\caption{The initial state of the layered moving knife algorithm is illustrated in this figure.}\label{figure:initial2}
\end{figure}

We  now present our layered moving knife algorithm. You can find a pseudocode of this algorithm in Algorithm~\ref{algorithm:layered}.  In this algorithm, we have three sets of agents: 
\begin{enumerate}[i.] \item\textit{dead agents} denoted by set $D$, \item \textit{agents in progress} denoted by set $P$, \item \textit{agents in the queue} denoted by set $Q$.
\end{enumerate}
 Dead agents are always a suffix of the sequence $\agent_{1},\agent_{2},\allowbreak\ldots,\agent_n$ and agents in the queue are a prefix of $\agent_{1},\agent_{2},\allowbreak\ldots,\agent_n$. Agents in between them are in progress. We assign the chores to agents one by one from $\chore_m$ to $\chore_1$. Thus, we introduce a pointer $s$ to the next chore to be assigned, initially equal to $m$. In the beginning, all agents are in the queue except agent $\agent_n$ who is in progress and none of the agents is dead.  In other words, we have $D = \emptyset$, $P = \{\agent_n\}$, and  $Q = \{\agent_1, \agent_2,\allowbreak \ldots, \agent_{n-1}\}$, see Figure \ref{figure:initial2} for illustration.

Our algorithm consists of multiple rounds until all chores are assigned. We proceed so that the following conditions hold at the beginning of each round; we refer to these conditions as \textit{safety measures}:
\begin{itemize}
\item $\sum_{\agent_i\in P}\entitlement_i\geq\tfrac{4}{5}\cdot \sum_{\agent_i\in D}\entitlement_i$.
\item If $\minp$ is the smallest index of an agent in $P$, then
$m-s\geq\tfrac{\sum_{j>\minp}\entitlement_j}{\entitlement_{\minp}}$.
\end{itemize}
Before we move any further, we would like to clarify what each of the conditions means in our algorithm. The first condition simply means that the total weight of the agents in progress is at least $\tfrac45$ of the total weight of the dead agents. The second condition implies a bound on the cost of each of the remaining chores to the agents. We show in the following that, as long as there is a remaining chore, the second condition implies that for each agent $\agent_i\in\agents$, $\cost_i(\{\chore_s\})\leq\entitlement_{\minp}$ holds.

\begin{lemma}\label{lemma:3}
If $m-s\geq\tfrac{\sum_{j>\minp}\entitlement_j}{\entitlement_{\minp}}$ and $s\geq1$, then
$\cost_i(\{\chore_s\})\leq\entitlement_{\minp}$ for every agent $\agent_i$.
\end{lemma}
\begin{proof}
	Fix an agent $\agent_i$ and a partition
	$\langle B_1,\allowbreak\ldots,B_n\rangle$ such that
	$\cost_i(B_j)\leq\entitlement_j$ for every $j\in[n]$.
	Suppose, toward a contradiction, that
	$\cost_i(\{\chore_s\})>\entitlement_{\minp}$.
	Since the chores are sorted, each of the $m-s+1$ chores
	$\chore_s,\allowbreak\ldots,\chore_m$ has cost greater than
	$\entitlement_{\minp}$.
		Since $\entitlement_j\leq\entitlement_{\minp}$ for every
	$j\leq\minp$, none of these chores can belong to
	$B_1,\allowbreak\ldots,B_{\minp}$. Hence, they must all be contained in
	$B_{\minp+1},\allowbreak\ldots,B_n$. On the one hand, by the second safety measure, their total cost is greater than
	
	$$
	(m-s+1)\entitlement_{\minp}
	>
	\sum_{j>\minp}\entitlement_j.
	$$
	
	On the other hand, since $\cost_i(B_j)\leq\entitlement_j$, the total
	cost of $B_{\minp+1},\allowbreak\ldots,B_n$ is at most
	$\sum_{j>\minp}\entitlement_j$. This is a contradiction.
\end{proof}

Note that, since we assume that chores are sorted, Lemma~\ref{lemma:3} gives the same bound for every remaining chore $\chore_1,\allowbreak\ldots,\chore_s$. Both safety measures hold initially, since $P=\{\agent_n\}$ and $D=\emptyset$.

At the beginning of a round, let $\minp'$ be the largest index $i<\minp$ such that
\[
\sum_{i\leq j<\minp}\entitlement_j
\geq\frac45\sum_{\minp\leq j\leq n}\entitlement_j.
\]
If no index satisfies this condition, we set $\minp'=1$. We call a round with $\minp'=1$ a \textit{terminal round}. Following the maximality of $\minp'$ when the defining condition is satisfied, or the failure of that condition otherwise, we have
\begin{equation}\label{ineq:smallest}
\sum_{\minp'<j<\minp}\entitlement_j
<\frac45\sum_{\minp\leq j\leq n}\entitlement_j.
\end{equation}
Lemma~\ref{lemma:4} shows that a terminal round assigns all remaining chores. If chores remain after any other round, the agents of $P$ move to $D$, and agents $\agent_{\minp'},\allowbreak\ldots,\agent_{\minp-1}$ form the next progress set. The first safety measure then follows directly from the choice of $\minp'$.

\begin{figure}[t]

\begin{center}

\tikzset{every picture/.style={line width=0.75pt}} 
\scalebox{0.5}{
\begin{tikzpicture}[x=0.75pt,y=0.75pt,yscale=-1,xscale=1]
	
	\draw   (38,109.5) -- (76.5,109.5) -- (76.5,148) -- (38,148) -- cycle ;
	\draw   (88,109.5) -- (126.5,109.5) -- (126.5,148) -- (88,148) -- cycle ;
	\draw   (208,109.5) -- (246.5,109.5) -- (246.5,148) -- (208,148) -- cycle ;
	\draw   (258,109.5) -- (296.5,109.5) -- (296.5,148) -- (258,148) -- cycle ;
	\draw   (418,109.5) -- (456.5,109.5) -- (456.5,148) -- (418,148) -- cycle ;
	\draw   (468,109.5) -- (506.5,109.5) -- (506.5,148) -- (468,148) -- cycle ;
	\draw   (518,109.5) -- (556.5,109.5) -- (556.5,148) -- (518,148) -- cycle ;
	\draw   (568,109.5) -- (606.5,109.5) -- (606.5,148) -- (568,148) -- cycle ;
	\draw    (326.5,82) -- (326.5,107) ;
	\draw [shift={(326.5,110)}, rotate = 270] [fill={rgb, 255:red, 0; green, 0; blue, 0 }  ][line width=0.08]  [draw opacity=0] (10.72,-5.15) -- (0,0) -- (10.72,5.15) -- (7.12,0) -- cycle    ;
	\draw   (28,249.5) -- (66.5,249.5) -- (66.5,288) -- (28,288) -- cycle ;
	\draw   (78,249.5) -- (116.5,249.5) -- (116.5,288) -- (78,288) -- cycle ;
	\draw   (305,249.5) -- (343.5,249.5) -- (343.5,288) -- (305,288) -- cycle ;
	\draw   (355,249.5) -- (393.5,249.5) -- (393.5,288) -- (355,288) -- cycle ;
	\draw   (448,249.5) -- (486.5,249.5) -- (486.5,288) -- (448,288) -- cycle ;
	\draw   (498,249.5) -- (536.5,249.5) -- (536.5,288) -- (498,288) -- cycle ;
	\draw   (548,249.5) -- (586.5,249.5) -- (586.5,288) -- (548,288) -- cycle ;
	\draw   (598,249.5) -- (636.5,249.5) -- (636.5,288) -- (598,288) -- cycle ;
	\draw  [fill={rgb, 255:red, 248; green, 231; blue, 28 }  ,fill opacity=0.32 ] (18.5,234) -- (211.5,234) -- (211.5,300) -- (18.5,300) -- cycle ;
	\draw  [fill={rgb, 255:red, 80; green, 227; blue, 194 }  ,fill opacity=0.34 ] (211.5,234) -- (349.5,234) -- (349.5,300) -- (211.5,300) -- cycle ;
	\draw   (308,109.5) -- (346.5,109.5) -- (346.5,148) -- (308,148) -- cycle ;
	\draw    (230.5,82) -- (230.5,107) ;
	\draw [shift={(230.5,110)}, rotate = 270] [fill={rgb, 255:red, 0; green, 0; blue, 0 }  ][line width=0.08]  [draw opacity=0] (10.72,-5.15) -- (0,0) -- (10.72,5.15) -- (7.12,0) -- cycle    ;
	\draw  [fill={rgb, 255:red, 184; green, 233; blue, 134 }  ,fill opacity=0.31 ] (199.5,94) -- (356.5,94) -- (356.5,182) -- (199.5,182) -- cycle ;
	\draw   (215,249.5) -- (253.5,249.5) -- (253.5,288) -- (215,288) -- cycle ;
	\draw   (169,249.5) -- (207.5,249.5) -- (207.5,288) -- (169,288) -- cycle ;
	\draw  [fill={rgb, 255:red, 155; green, 155; blue, 155 }  ,fill opacity=0.27 ] (350.5,234) -- (644.5,234) -- (644.5,300) -- (350.5,300) -- cycle ;
	
	\draw (50,119.4) node [anchor=north west][inner sep=0.75pt]    {$b_{1}$};
	\draw (100,119.4) node [anchor=north west][inner sep=0.75pt]    {$b_{2}$};
	\draw (265,119.4) node [anchor=north west][inner sep=0.75pt]    {$b_{s-1}$};
	\draw (139,124.4) node [anchor=north west][inner sep=0.75pt]  [font=\Huge]  {$\dotsc $};
	\draw (422,119.4) node [anchor=north west][inner sep=0.75pt]    {$b_{m-3}$};
	\draw (471,118.4) node [anchor=north west][inner sep=0.75pt]    {$b_{m-2}$};
	\draw (521,118.4) node [anchor=north west][inner sep=0.75pt]    {$b_{m-1}$};
	\draw (578,118.4) node [anchor=north west][inner sep=0.75pt]    {$b_{m}$};
	\draw (301,62.4) node [anchor=north west][inner sep=0.75pt]    {$s=\ktwo$};
	\draw (40,265.4) node [anchor=north west][inner sep=0.75pt]    {$a_{1}$};
	\draw (90,265.4) node [anchor=north west][inner sep=0.75pt]    {$a_{2}$};
	\draw (311,265.4) node [anchor=north west][inner sep=0.75pt]    {$a_{x-1}$};
	\draw (365,265.4) node [anchor=north west][inner sep=0.75pt]    {$a_{x}$};
	\draw (395,270.4) node [anchor=north west][inner sep=0.75pt]  [font=\Huge]  {$\dotsc $};
	\draw (452,265.4) node [anchor=north west][inner sep=0.75pt]    {$a_{n-3}$};
	\draw (501,265.4) node [anchor=north west][inner sep=0.75pt]    {$a_{n-2}$};
	\draw (551,265.4) node [anchor=north west][inner sep=0.75pt]    {$a_{n-1}$};
	\draw (608,265.4) node [anchor=north west][inner sep=0.75pt]    {$a_{n}$};
	\draw (18,217) node [anchor=north west][inner sep=0.75pt]   [align=left] {$Q$};
	\draw (213,213) node [anchor=north west][inner sep=0.75pt]   [align=left] {$P$};
	\draw (369,124.4) node [anchor=north west][inner sep=0.75pt]  [font=\Huge]  {$\dotsc $};
	\draw (321,119.4) node [anchor=north west][inner sep=0.75pt]    {$b_{s}$};
	\draw (215,119.4) node [anchor=north west][inner sep=0.75pt]    {$b_{s-2}$};
	\draw (219,62.4) node [anchor=north west][inner sep=0.75pt]    {$\kone$};
	\draw (235,160) node [anchor=north west][inner sep=0.75pt]   [align=left] {knife interval};
	\draw (257,270.4) node [anchor=north west][inner sep=0.75pt]  [font=\Huge]  {$\dotsc $};
	\draw (216,265.4) node [anchor=north west][inner sep=0.75pt]    {$a_{\minp}$};
	\draw (122,270.4) node [anchor=north west][inner sep=0.75pt]  [font=\Huge]  {$\dotsc $};
	\draw (168,265.4) node [anchor=north west][inner sep=0.75pt]  [font=\footnotesize]  {$a_{\minp-1}$};
	\draw (353,213) node [anchor=north west][inner sep=0.75pt]   [align=left] {$D$};

\end{tikzpicture}
}

\end{center}

	\caption{The layered moving knife algorithm is illustrated in this figure. $\agent_x$ is the lowest-index agent in $D$.
	}\label{figure:initial}
\end{figure}

We assign chores only to agents of $P$ in each round. For each $\agent_i\in P$, put $\lfloor3\entitlement_i/\entitlement_{\minp}\rfloor$ copies of this agent in a multiset $P'$. To form a bundle, we place two knives at the next chore to be assigned, $\chore_s$, and denote this position by $\ktwo=s$. We keep the first knife fixed and move the second knife to the left, denoting its position by $\kone$. The two knives delimit the bundle $\{\chore_{\kone},\allowbreak\ldots,\chore_{\ktwo}\}$, including both endpoints. Initially, $\kone=\ktwo$, so the bundle contains only $\chore_s$; Lemma~\ref{lemma:3} guarantees that $\cost_i(\{\chore_s\})\leq\entitlement_{\minp}$ for every agent $\agent_i$. We move the second knife one chore to the left as long as $\kone>1$ and the enlarged bundle $\{\chore_{\kone-1},\allowbreak\ldots,\chore_{\ktwo}\}$ costs at most $\tfrac{18}{5}\cdot\entitlement_{\minp}$ to some agent remaining in $P'$. Once no further move is possible, we assign the bundle between the knives to any remaining agent whose cost for it is at most this threshold, and remove one copy of that agent from $P'$. We then set $s=\kone-1$. If chores and copies both remain, we place both knives at $\chore_s$ and repeat the procedure; otherwise, the round ends.

Consider a round that begins with the safety measures satisfied. If all chores have been assigned by the end of the round, the algorithm terminates and there is nothing further to prove about the next round. Otherwise, at least one chore remains unassigned. In Lemma~\ref{lemma:round-progress}, we consider a round after which at least one chore remains unassigned. Let $P$, $D$, and $s$ denote the progress set, the dead set, and the pointer at the beginning of this round, respectively. Let $s'$ denote the value of the pointer at the end of the round. Thus, $m-s'$ chores have been assigned through the end of the round, of which $s-s'$ were assigned during this round. Since the round has removed all copies in $P'$, we fix $\agent_j$ to be the agent whose copy is the last to be removed. We first show that the chores assigned so far have large total cost to $\agent_j$, and then use this bound to establish the safety measures for the next round.

\begin{lemma}\label{lemma:round-progress}
The chores assigned through the end of the round have total cost to $\agent_j$ greater than
$
\frac{18}{5}\sum_{\agent_k\in P\cup D}\entitlement_k.
$
\end{lemma}
\begin{proof}
\begingroup
Let
$
c=\sum_{\agent_k\in P}\left\lfloor\frac{3\entitlement_k}{\entitlement_{\minp}}\right\rfloor
$
be the initial number of copies in $P'$. Recall that $\agent_j$ is the agent whose copy is the last to be removed from $P'$ in the current round. Thus, she has a copy in $P'$ whenever a bundle is formed. When the knife stops, adding the next chore would therefore exceed her threshold as well.

At the beginning of the round, $\chore_s$ is the largest remaining chore. By Lemma~\ref{lemma:3},
\begin{equation}\label{eq:round-largest}
0<\cost_j(\{\chore_s\})\leq\entitlement_{\minp}.
\end{equation}
Since chores are sorted, each of the $m-s$ chores assigned before this round costs $\agent_j$ at least $\cost_j(\{\chore_s\})$.

Next, consider a bundle assigned during the round. Agent $\agent_j$ still has a copy when the bundle is formed, and some chore remains after it is assigned. The knife stopped because adding the next chore would make the bundle cost more than $\tfrac{18}{5}\cdot\entitlement_{\minp}$ to every agent still represented in $P'$, including $\agent_j$. That next chore costs $\agent_j$ at most $\cost_j(\{\chore_s\})$. Thus each of the $c$ bundles assigned during the round costs her more than
$
\frac{18}{5}\entitlement_{\minp}-\cost_j(\{\chore_s\}).
$
Combining the chores assigned before the round with the bundles assigned during it, their total cost is strictly greater than
\begin{align}
&\begin{aligned}&(m-s)\cost_j(\{\chore_s\})
+c\left(\frac{18}{5}\entitlement_{\minp}-\cost_j(\{\chore_s\})\right)\\&=\frac{18}{5}c\entitlement_{\minp}
+(m-s-c)\cost_j(\{\chore_s\}).\end{aligned}\label{eq:round-cost}
\end{align}

It remains to express this cost bound in terms of entitlements.  First, we have
\begingroup
\settowidth{\chainlhswidth}{$\displaystyle \sum_{\agent_k\in P\cup D}\entitlement_k$}
\begin{align}
\sum_{\agent_k\in P\cup D}\entitlement_k&=\sum_{\agent_k\in P}\entitlement_k+\sum_{\agent_k\in D}\entitlement_k\notag\\&\annotatedrow{\leq\sum_{\agent_k\in P}\entitlement_k+\frac54\sum_{\agent_k\in P}\entitlement_k}{\text{first safety measure}}\notag\\&=\frac94\sum_{\agent_k\in P}\entitlement_k\label{eq:round-safety}
\end{align}
\endgroup
The second safety measure, multiplied by $\entitlement_{\minp}$, gives
\begingroup
\settowidth{\chainlhswidth}{$\displaystyle (m-s)\entitlement_{\minp}$}
\begin{align}
(m-s)\entitlement_{\minp}&\annotatedrow{\geq\sum_{k>\minp}\entitlement_k}{\text{second safety measure}}\notag\\&\annotatedrow{=\sum_{\agent_k\in P\cup D}\entitlement_k-\entitlement_{\minp}}{P\cup D=\{\agent_{\minp},\allowbreak\ldots,\agent_n\}}\label{eq:round-count}
\end{align}
\endgroup

We also need a lower bound on the number of bundles assigned in this round. Recall that $\entitlement_k\geq\entitlement_{\minp}$ for every $\agent_k\in P$. Hence $\lfloor3\entitlement_k/\entitlement_{\minp}\rfloor\geq3$, and
\begingroup
\settowidth{\chainlhswidth}{$\displaystyle \frac{3\entitlement_k}{\entitlement_{\minp}}$}
$
\frac{3\entitlement_k}{\entitlement_{\minp}}<\left\lfloor\frac{3\entitlement_k}{\entitlement_{\minp}}\right\rfloor+1\leq\frac43\left\lfloor\frac{3\entitlement_k}{\entitlement_{\minp}}\right\rfloor
$.
\endgroup
Rearranging gives $\lfloor3\entitlement_k/\entitlement_{\minp}\rfloor\geq\tfrac94\cdot\entitlement_k/\entitlement_{\minp}$. Agent $\agent_{\minp}$ contributes exactly three copies, so
\begingroup
\settowidth{\chainlhswidth}{$\displaystyle c\entitlement_{\minp}$}
\begin{align}
c\entitlement_{\minp}&=3\entitlement_{\minp}
+\sum_{\substack{\agent_k\in P\\k\ne\minp}}
\left\lfloor\frac{3\entitlement_k}{\entitlement_{\minp}}\right\rfloor\entitlement_{\minp}\notag\\&\begin{aligned}&\geq3\entitlement_{\minp}
+\frac94\sum_{\substack{\agent_k\in P\\k\ne\minp}}\entitlement_k\end{aligned}\notag\\&=\frac94\sum_{\agent_k\in P}\entitlement_k+\frac34\entitlement_{\minp}\label{eq:round-copies}
\end{align}
\endgroup

The sign of $m-s-c$ determines which of the two contributions in~\eqref{eq:round-cost} is important.
If $m-s\geq c$, we have:
\begingroup
\setlength{\chainlhswidth}{0pt}
\begin{align*}
&\frac{18}{5}c\entitlement_{\minp}
+(m-s-c)\cost_j(\{\chore_s\})\\&\annotatedrow{\geq\frac{18}{5}c\entitlement_{\minp}}{m-s-c\geq0}\\&\annotatedrow{\geq\frac{81}{10}\sum_{\agent_k\in P}\entitlement_k
+\frac{27}{10}\entitlement_{\minp}}{\text{by }\eqref{eq:round-copies}}\\&\annotatedrow{\geq\frac{18}{5}\sum_{\agent_k\in P\cup D}\entitlement_k
+\frac{27}{10}\entitlement_{\minp}}{\text{by }\eqref{eq:round-safety}}\\&\annotatedrow{>\frac{18}{5}\sum_{\agent_k\in P\cup D}\entitlement_k}{\entitlement_{\minp}>0}
\end{align*}
\endgroup

If $m-s<c$, the coefficient of $\cost_j(\{\chore_s\})$ in~\eqref{eq:round-cost} is negative. We therefore use~\eqref{eq:round-largest} to obtain the following lower bound:
\begingroup
\begin{align*}
&\frac{18}{5}c\entitlement_{\minp}
+(m-s-c)\cost_j(\{\chore_s\})\\&\annotatedrow{\geq\frac{18}{5}c\entitlement_{\minp}
+(m-s-c)\entitlement_{\minp}}{\eqref{eq:round-largest},\ m-s-c<0}\\&=\frac{13}{5}c\entitlement_{\minp}+(m-s)\entitlement_{\minp}\\&\annotatedrow{\begin{aligned}&\geq\frac{117}{20}\sum_{\agent_k\in P}\entitlement_k
+\sum_{\agent_k\in P\cup D}\entitlement_k\\&\qquad+\frac{19}{20}\entitlement_{\minp}\end{aligned}}{\eqref{eq:round-copies}\text{ and }\eqref{eq:round-count}}\\&\annotatedrow{\geq\frac{18}{5}\sum_{\agent_k\in P\cup D}\entitlement_k
+\frac{19}{20}\entitlement_{\minp}}{\text{by }\eqref{eq:round-safety}}\\&\annotatedrow{>\frac{18}{5}\sum_{\agent_k\in P\cup D}\entitlement_k}{\entitlement_{\minp}>0}
\end{align*}
\endgroup

In either case, the expression in~\eqref{eq:round-cost} exceeds $\tfrac{18}{5}\cdot\sum_{\agent_k\in P\cup D}\entitlement_k$. 
\endgroup
\end{proof}

Lemma~\ref{lemma:4} uses the cost bound from Lemma~\ref{lemma:round-progress} to establish the safety measures for the next round.

\begin{lemma}\label{lemma:4}
If the current round is terminal, all remaining chores are assigned. Otherwise, either all chores are assigned or both safety measures hold for the next round.
\end{lemma}
\begin{proof}
\begingroup
Suppose first that this is a terminal round, so $\minp'=1$, but some chore remains unassigned. Recall that $\agent_j$ is the agent whose copy is the last to be removed from $P'$ in the current round. By Lemma~\ref{lemma:round-progress}, the total cost to $\agent_j$ of the chores assigned through the end of the round satisfies 
\begingroup
\begingroup
\setlength{\chainlhswidth}{0pt}
\begin{align*}
&\sum_{\chore_k\text{ assigned}}\cost_j(\{\chore_k\})\\&\annotatedrow{>\frac{18}{5}\sum_{\agent_k\in P\cup D}\entitlement_k}{\text{Lemma~\ref{lemma:round-progress}}}\\&=2\left(\frac95\sum_{\agent_k\in P\cup D}\entitlement_k\right)\\&\annotatedrow{>2\left(\sum_{1<k<\minp}\entitlement_k
+\sum_{\agent_k\in P\cup D}\entitlement_k\right)}{\eqref{ineq:smallest},\ \minp'=1}\\&\annotatedrow{\geq2(1-\entitlement_1)}{P\cup D=\{\agent_{\minp},\allowbreak\ldots,\agent_n\}}\\&\annotatedrow{\geq1}{\entitlement_1\leq\tfrac1n\leq\tfrac12.}
\end{align*}
\endgroup
\endgroup
 The last line uses that $\entitlement_1$ is the smallest entitlement and $n\geq2$. Thus the assigned chores alone cost $\agent_j$ more than $1$, contradicting $\cost_j(\chores)\leq1$. All chores must therefore be assigned in a terminal round.

Now suppose $\minp'>1$ and some chores remain. We will show that enough chores have been assigned for the second safety measure to hold on the next round. To use Lemma~\ref{lemma:round-progress} for this purpose, we first relate the weight $\sum_{\agent_k\in P\cup D}\entitlement_k$ in its bound to the weight $\sum_{k>\minp'}\entitlement_k$ appearing in that safety measure:
\begingroup
\begingroup
\settowidth{\chainlhswidth}{$\displaystyle \sum_{k>\minp'}\entitlement_k$}
\begin{align}
\sum_{k>\minp'}\entitlement_k&\annotatedrow{=\sum_{\minp'<k<\minp}\entitlement_k
+\sum_{\agent_k\in P\cup D}\entitlement_k}{\text{splitting at }\minp}\notag\\&\annotatedrow{<\frac45\sum_{\agent_k\in P\cup D}\entitlement_k
+\sum_{\agent_k\in P\cup D}\entitlement_k}{\text{by }\eqref{ineq:smallest}}\notag\\&=\frac95\sum_{\agent_k\in P\cup D}\entitlement_k\label{eq:round-tail}
\end{align}
\endgroup
\endgroup
Consequently, the cost of the chores assigned through the end of the round satisfies
\begingroup
\begingroup
\setlength{\chainlhswidth}{0pt}
\begin{align*}
&\sum_{\chore_k\text{ assigned}}\cost_j(\{\chore_k\})\\&\annotatedrow{>\frac{18}{5}\sum_{\agent_k\in P\cup D}\entitlement_k}{\text{Lemma~\ref{lemma:round-progress}}}\\&=2\left(\frac95\sum_{\agent_k\in P\cup D}\entitlement_k\right)\\&\annotatedrow{>2\sum_{k>\minp'}\entitlement_k}{\text{by }\eqref{eq:round-tail}.}
\end{align*}
\endgroup
\endgroup

Recall that $s'$ is the value of the pointer at the end of the current round, so $m-s'$ chores have been assigned by that point. Suppose that the second safety measure fails for the next round. Then
\begin{equation}\label{eq:round-failed-safety}
m-s'
<\frac{\sum_{k>\minp'}\entitlement_k}{\entitlement_{\minp'}}.
\end{equation}
Fix a partition $\langle B_1,\allowbreak\ldots,B_n\rangle$ with $\cost_j(B_k)\leq\entitlement_k$ for every $k$, as guaranteed by the normalization.

For the assigned chores in $B_{\minp'+1},\allowbreak\ldots,B_n$, we have
\begingroup
\setlength{\chainlhswidth}{0pt}
\begin{align}
&\sum_{k>\minp'}\ \sum_{\substack{\chore\in B_k\\\chore\text{ assigned}}}\cost_j(\{\chore\})\notag\\&\annotatedrow{\leq\sum_{k>\minp'}\cost_j(B_k)}{\text{nonnegative additive costs}}\notag\\&\annotatedrow{\leq\sum_{k>\minp'}\entitlement_k}{\cost_j(B_k)\leq\entitlement_k.}\label{eq:round-high-cost}
\end{align}
\endgroup
For $\chore\in B_k$ with $k\leq\minp'$, we instead use the bound on each individual chore:
\begingroup
\settowidth{\chainlhswidth}{$\displaystyle \cost_j(\{\chore\})$}
\begin{align}
\cost_j(\{\chore\})&\annotatedrow{\leq\cost_j(B_k)}{\chore\in B_k}\notag\\&\annotatedrow{\leq\entitlement_k}{\text{choice of the partition}}\notag\\&\annotatedrow{\leq\entitlement_{\minp'}}{k\leq\minp'.}\label{eq:round-chore-cost}
\end{align}
\endgroup
Consequently,
\begingroup
\begingroup
\setlength{\chainlhswidth}{0pt}
\begin{align}
&\sum_{k\leq\minp'}\ \sum_{\substack{\chore\in B_k\\\chore\text{ assigned}}}\cost_j(\{\chore\})\notag\\&\annotatedrow{\leq\entitlement_{\minp'}(m-s')}{\text{by }\eqref{eq:round-chore-cost}}\notag\\&\annotatedrow{<\entitlement_{\minp'}
\frac{\sum_{k>\minp'}\entitlement_k}{\entitlement_{\minp'}}}{\text{by }\eqref{eq:round-failed-safety}}\notag\\&=\sum_{k>\minp'}\entitlement_k.\label{eq:round-low-cost}
\end{align}
\endgroup
\endgroup
Since the bundles partition $\chores$, every assigned chore belongs to exactly one of them. Adding~\eqref{eq:round-high-cost} and~\eqref{eq:round-low-cost} therefore gives
\begingroup
\begingroup
\setlength{\chainlhswidth}{0pt}
\begin{align*}
&\sum_{\chore\text{ assigned}}\cost_j(\{\chore\})\\&\annotatedrow{=\sum_{k=1}^n\ \sum_{\substack{\chore\in B_k\\\chore\text{ assigned}}}\cost_j(\{\chore\})}{\text{the bundles partition }\chores}\\&\annotatedrow{<\sum_{k>\minp'}\entitlement_k+\sum_{k>\minp'}\entitlement_k}{\eqref{eq:round-high-cost},\eqref{eq:round-low-cost}}\\&=2\sum_{k>\minp'}\entitlement_k.
\end{align*}
\endgroup
\endgroup
This contradicts the lower bound obtained from Lemma~\ref{lemma:round-progress} and~\eqref{eq:round-tail}. Hence~\eqref{eq:round-failed-safety} is false, and the second safety measure holds for the next round.

For the first safety measure, the next progress set is $\{\agent_{\minp'},\allowbreak\ldots,\agent_{\minp-1}\}$ and the next dead set is the current $P\cup D$. Their weights satisfy
\begingroup
\settowidth{\chainlhswidth}{$\displaystyle \sum_{\minp'\leq k<\minp}\entitlement_k$}
\begin{align*}
\sum_{\minp'\leq k<\minp}\entitlement_k&\annotatedrow{\geq\frac45\sum_{\minp\leq k\leq n}\entitlement_k}{\text{choice of }\minp'}\\&\annotatedrow{=\frac45\sum_{\agent_k\in P\cup D}\entitlement_k}{P\cup D=\{\agent_{\minp},\allowbreak\ldots,\agent_n\}.}
\end{align*}
\endgroup
This is precisely the first safety measure for the next round.
\endgroup
\end{proof}

\begin{algorithm}[t]
\caption{Layered moving knife algorithm}\label{algorithm:layered}
\begin{algorithmic}[1]
\STATE \textbf{Input:} sorted chores, ordered entitlements, and costs with $\WMMS_i\leq\entitlement_i$
\STATE \textbf{Output:} $A_1,\allowbreak\ldots,A_n$
\STATE $s\leftarrow m$; $D\leftarrow\emptyset$; $P\leftarrow\{\agent_n\}$; $Q\leftarrow\agents\setminus P$
\STATE $A_i\leftarrow\emptyset$ for each $i\in[n]$
\WHILE{$s>0$}
\STATE $\minp\leftarrow$ smallest index of an agent in $P$
\STATE $\minp'\leftarrow$ largest $i<\minp$ with $\sum_{i\leq j<\minp}\entitlement_j\geq\tfrac{4}{5}\cdot \sum_{\agent_j\in D\cup P}\entitlement_j$, or $1$ if none exists
\STATE $q_i\leftarrow\lfloor3\entitlement_i/\entitlement_{\minp}\rfloor$ for every $\agent_i\in P$
\WHILE{$s>0$ and some $q_i>0$}
\STATE $(\kone,\ktwo)\leftarrow(s,s)$
\WHILE{$\kone>1$ and some $\agent_i\in P$ with $q_i>0$ satisfies $\cost_i(\{\chore_{\kone-1},\allowbreak\ldots,\chore_{\ktwo}\})\leq\tfrac{18}{5}\cdot \entitlement_{\minp}$}
\STATE $\kone\leftarrow\kone-1$
\ENDWHILE
\STATE Choose $\agent_i\in P$ with $q_i>0$ and $\cost_i(\{\chore_{\kone},\allowbreak\ldots,\chore_{\ktwo}\})\leq\tfrac{18}{5}\cdot \entitlement_{\minp}$
\STATE $A_i\leftarrow A_i\cup\{\chore_{\kone},\allowbreak\ldots,\chore_{\ktwo}\}$; $q_i\leftarrow q_i-1$; $s\leftarrow\kone-1$
\ENDWHILE
\IF{$s=0$}
\STATE \textbf{return} $A_1,\allowbreak\ldots,A_n$
\ENDIF
\STATE $D\leftarrow D\cup P$ \COMMENT{Lemma~\ref{lemma:4} guarantees $\minp'<\minp$ here}
\STATE $P\leftarrow\{\agent_{\minp'},\allowbreak\ldots,\agent_{\minp-1}\}$; $Q\leftarrow\{\agent_1,\allowbreak\ldots,\agent_{\minp'-1}\}$
\ENDWHILE
\RETURN $A_1,\allowbreak\ldots,A_n$
\end{algorithmic}
\end{algorithm}

\Needspace{9\baselineskip}
\begin{theorem}\label{theorem:main}
The asymmetric chore division problem admits a $\tfrac{54}{5}$-$\WMMS$ assignment.
\end{theorem}
\begin{proof}
Lemma~\ref{lemma:4} preserves the safety measures until all chores are assigned. Every unfinished round strictly decreases the minimum index of the progress set, so the algorithm terminates. Each agent belongs to $P$ in only one round. In that round, she receives at most $\lfloor3\entitlement_i/\entitlement_{\minp}\rfloor$ bundles, each of cost at most $\tfrac{18}{5}\cdot \entitlement_{\minp}$. Thus her total cost is at most
$
\left\lfloor\frac{3\entitlement_i}{\entitlement_{\minp}}\right\rfloor
\frac{18}{5}\entitlement_{\minp}\leq\frac{54}{5}\entitlement_i.
$
\end{proof}

\paragraph{A polynomial-time guarantee.}
Theorem~\ref{theorem:main} can be implemented in polynomial time with an arbitrarily small loss in the approximation guarantee by approximating each $\WMMS_i$ via a PTAS.
\begin{corollary}\label{cor:ptas-implementation}
For every fixed $\epsilon>0$, a $(\tfrac{54}{5}+\epsilon)$-$\WMMS$ assignment can be computed in polynomial time for rational costs and entitlements.
\end{corollary}
\begin{proof}
Fix an agent $\agent_i$. Regard each chore $\chore$ as a job of size $\cost_i(\{\chore\})$, and introduce one machine of speed $\entitlement_j$ for each $j\in[n]$. The minimum makespan equals
\[
\min_{\langle B_1,\allowbreak\ldots,B_n\rangle\in\Pi_n(\chores)}
\max_j\frac{\cost_i(B_j)}{\entitlement_j}
=\frac{\WMMS_i}{\entitlement_i}.
\]
The PTAS of Hochbaum and Shmoys~\cite{hochbaum1988uniform} for machines of different speeds returns a feasible schedule within a factor $1+\delta$ of this optimum for any fixed $\delta>0$. Multiplying its makespan by $\entitlement_i$ gives an upper estimate
$\WMMS_i\leq\widehat\mu_i\leq(1+\delta)\WMMS_i$.
For each agent $\agent_i$, scale her cost function by $\entitlement_i/\widehat\mu_i$. The schedule returned for this agent defines a partition whose scaled bundle costs are at most $\entitlement_j$ for every $j$, so the analysis of Theorem~\ref{theorem:main} applies. After undoing the scaling, the cost of her assigned bundle is at most
\[
\frac{54}{5}\widehat\mu_i\leq\frac{54}{5}(1+\delta)\WMMS_i.
\]
Set $\delta=5\epsilon/54$. The PTAS runs in polynomial time for fixed $\epsilon$, and the assignment procedure admits the polynomial-time implementation described in Section~\ref{sec:comp}.
\end{proof}

\section[Almost-Linear-Time Implementation of the Layered Moving Knife]{Almost-Linear-Time Implementation of the Layered Moving Knife}\label{sec:comp}

In this section, we show that allowing an additional constant factor in the approximation enables the layered moving-knife algorithm to run in almost linear time.\footnote{The input specifies all valuation functions, so the input size is $\Omega(mn)$.} As mentioned earlier, computing exact weighted maximin share values is NP-hard \cite{Procaccia:first}. In Lemma~\ref{lem:app2}, we show how to compute a polynomial-time constant approximation of $\WMMS_i$ for each agent $\agent_i$.

\begin{lemma}\label{lem:app2}
	Fix an agent $\agent_i$ and $0<\epsilon<1$. We can find a partition $A = \langle \alloc_1,\alloc_2,\allowbreak\ldots,\alloc_n\rangle$ in time $O(\tfrac{1}{\epsilon}\cdot m \log n + m \log m)$ such that
	$ \max_{\agent_j \in \agents} \cost_i(\alloc_j)\cdot\tfrac{\entitlement_i}{\entitlement_j} \leq 2(1+\epsilon)\cdot\WMMS_i.$
\end{lemma}
\begin{proof}
If $n=1$, assigning all chores to the only agent gives an exact weighted maximin share partition. Henceforth, assume $n\geq2$. Fix an agent $\agent_i$. We first establish a lower bound and an upper bound on $\WMMS_i$.
\begingroup
For the lower bound, every partition $\langle A_1,\allowbreak\ldots,A_n\rangle$ satisfies
\begin{align*}
\max_{j\in[n]}\frac{\cost_i(A_j)}{\entitlement_j}
&\geq\sum_{j=1}^n\entitlement_j\frac{\cost_i(A_j)}{\entitlement_j}
&\quad \sum_{j=1}^n\entitlement_j=1\\
&=\cost_i(\chores)
&.
\end{align*}
Multiplying by $\entitlement_i$ and taking the minimum over all partitions gives
\begin{equation}\label{eq:lb}
\WMMS_i \;\ge\; \entitlement_i \cdot \cost_i(\chores).
\end{equation}
\endgroup

For the upper bound, consider the partition in which all chores are assigned to the bundle corresponding to the largest entitlement $\entitlement_n$. Since this is a feasible partition, the resulting upper bound on $\WMMS_i$ under this partition is
\begin{equation}\label{eq:ub}
\WMMS_i \;\le\; \frac{\entitlement_i}{\entitlement_n}\,\cost_i(\chores).
\end{equation}

Since $\entitlement_n$ is the largest entitlement, we have $\entitlement_n \ge 1/n$. Therefore, the ratio between the upper bound (Inequality \eqref{eq:ub}) and the lower bound (Inequality \eqref{eq:lb}) is $O(n)$.

Now suppose we are given an estimate $\tau$ of $\WMMS_i$. We show that if $\tau \ge \WMMS_i$, then the following simple algorithm produces a partition
$
A = \langle A_1, A_2,\allowbreak \ldots, A_n \rangle
$
satisfying
\begin{equation} \label{ineq:main}
\max_{\agent_j \in \agents} \cost_i(A_j)\,\frac{\entitlement_i}{\entitlement_j} \;\le\; 2\tau.
\end{equation}

\begingroup
Let $h=n$ denote the index of the next agent to receive a bundle, and let $r = m$ denote the index of the last unallocated chore. Initialize $A_j=\emptyset$ for all $j\in[n]$.
While $r \ge 1$ and $h \ge 1$, do the following:
\begin{itemize}
	\item Find the largest index $1 \le \ell \le r$ such that
	\[
	\cost_i(\{\chore_\ell, \chore_{\ell+1},\allowbreak \ldots, \chore_r\})
	\;\ge\;
	\frac{\entitlement_h}{\entitlement_i}\,\tau.
	\]
	If no such index exists, assign all remaining chores to $\agent_h$, set $r=0$, and terminate the loop.
	\item Assign the bundle $\{\chore_\ell, \chore_{\ell+1},\allowbreak \ldots, \chore_r\}$ to agent $\agent_h$, that is, set
	$A_h = \{\chore_\ell, \chore_{\ell+1},\allowbreak \ldots, \chore_r\}$.
	\item Update the pointers by setting $h = h - 1$ and $r = \ell - 1$.
\end{itemize}

If the loop terminates with $r\ge 1$, report failure; otherwise, return the partition $A$.

We argue that the algorithm does not run out of agents while chores remain unallocated.
Suppose, toward a contradiction, that the algorithm terminates with $h=0$ while at least one chore remains unallocated.
\endgroup
For every allocated bundle $A_j$, the construction of the algorithm guarantees
\[
\cost_i(A_j) \;\ge\; \frac{\entitlement_j}{\entitlement_i}\,\tau
\;\ge\; \frac{\entitlement_j}{\entitlement_i}\,\WMMS_i ,
\]
where the second inequality follows from $\tau \ge \WMMS_i$.
Summing over all such bundles, we obtain
\begin{align*}
	\cost_i(\chores)>\sum_{j=1}^n \cost_i(A_j)
	&\ge \sum_{j=1}^n \frac{\entitlement_j}{\entitlement_i}\,\WMMS_i \\
	&\ge\; \frac{\WMMS_i}{\entitlement_i} .
\end{align*}
However, this contradicts Inequality \eqref{eq:lb}.

We now prove that Inequality~\eqref{ineq:main} holds for the resulting partition. Suppose, for the sake of contradiction, that there exists an agent $\agent_j$ such that
\[
\cost_i(A_j)\,\frac{\entitlement_i}{\entitlement_j} \;>\; 2\tau .
\]

By the construction of the algorithm and since chores are processed in non-increasing order of cost, this implies that the bundle $A_j$ consists of a single chore.
Let $A_j = \{\chore_k\}$, where
\[
\cost_i(\{\chore_k\}) \;>\; 2\,\frac{\entitlement_j}{\entitlement_i}\,\tau .
\]

Consider the set of chores
\[
C' = \{\chore_k, \chore_{k+1},\allowbreak \ldots, \chore_m\},
\]
Note that in the optimal partitioning of agent $\agent_i$, all these chores must be allocated to agents
$\agent_{j+1},\allowbreak \ldots, \agent_n$; otherwise, the normalized cost of at least one agent would exceed $\WMMS_i$.
Indeed, suppose that some chore in $C'$ is assigned to an agent $\agent_t$ with $t\le j$.
Since costs are non-decreasing in the chore index, every chore in $C'$ has cost at least $\cost_i(\{\chore_k\})$, and we have
\[
\cost_i(\{\chore_k\}) \;>\; 2\,\frac{\entitlement_j}{\entitlement_i}\,\tau
\;\ge\;
\frac{\entitlement_t}{\entitlement_i}\,\WMMS_i ,
\]
where the last inequality follows from $\tau \ge \WMMS_i$ and $\entitlement_t \le \entitlement_j$.
Thus, assigning even a single chore from $C'$ to agent $\agent_t$ would result in a normalized cost strictly larger than $\WMMS_i$, which is a contradiction. 

If $j=n$, no agent can receive $\chore_k$ in an optimal partition, which is already a contradiction. We therefore assume $j<n$. Finally, note that the total cost of the chores in $C'$ satisfies
\[
\cost_i(C') \;\ge\; \sum_{t=j+1}^{n} \frac{\entitlement_t}{\entitlement_i}\,\tau \;+\; \cost_i(\{\chore_k\}) .
\]
This implies that, even if the chores in $C'$ were divisible and could be perfectly allocated among agents $\agent_{j+1},\allowbreak\ldots,\agent_n$ in proportion to their entitlements, the resulting normalized cost for some agent would be at least
\[
 \tau + \frac{\entitlement_i\,\cost_i(\{\chore_k\})}{\sum_{t=j+1}^{n}\entitlement_t} \;>\; \WMMS_i ,
\]
contradicting the definition of $\WMMS_i$.

Therefore, it remains to identify a value of $\tau$ for which the algorithm successfully returns the desired partition, while ensuring that $\tau$ is not significantly larger than $\WMMS_i$.
Fix a constant $0<\epsilon<1$. To find a proper assignment, we initialize 
\[
\tau = \entitlement_i \cdot \cost_i(\chores),
\] 
which serves as a lower bound on $\WMMS_i$. We then run the algorithm with this value of $\tau$. 

\begin{itemize}
	\item If the algorithm succeeds in producing a desired partition satisfying Inequality \eqref{ineq:main} for all the bundles, we output it as the final assignment.  
	\item Otherwise, we increase $\tau$ by a multiplicative factor of $(1+\epsilon)$, that is, set $\tau \gets (1+\epsilon)\tau$, and repeat the procedure.
\end{itemize}

The first time the algorithm returns a partition, we have
\[
\tau \le (1+\epsilon)\WMMS_i.
\] 
The resulting assignment then satisfies
\[
\max_{\agent_j \in \agents} \cost_i(A_j)\,\frac{\entitlement_i}{\entitlement_j} \;\le\; 2 \tau \;\le\; 2(1+\epsilon)\WMMS_i.
\] 
Therefore, the assignment is a $2(1+\epsilon)$-approximation of the optimal $\WMMS_i$ partition.

\paragraph{Time complexity.} 
Since $\tau$ grows geometrically by a factor of $(1+\epsilon)$ and starts from a lower bound of $\entitlement_i \cdot \cost_i(\chores)$, the number of iterations until $\tau$ first reaches or exceeds $\WMMS_i$ is at most
\[
1+\Big\lceil \log_{1+\epsilon} \frac{\WMMS_i}{\entitlement_i \cdot \cost_i(\chores)} \Big\rceil = O\Big(\frac{1}{\epsilon} \log n \Big).
\]
Each iteration runs in time $O(m)$ to construct the partition, where $m$ is the number of chores. Therefore, the overall running time of the algorithm is
$
O(\tfrac{m}{\epsilon}\cdot \log n ).
$
In addition, we assumed that the chores are sorted according to their cost for agent $\agent_i$, which also incurs an additional running time of $O(m \log m)$. The overall running time would be $$O\Big(\frac{m}{\epsilon} \log n + m \log m\Big).$$

\end{proof}

\begin{corollary}[of Lemma \ref{lem:app2}]\label{cor:1}
We can compute a $2(1+\epsilon)$-approximation of all $\WMMS_i$ in time
$
O({nm \log n}/{\epsilon} + nm \log m)
$.
\end{corollary}

Combining Corollary~\ref{cor:1}, the sorted-chores reduction, and the layered moving-knife implementation yields Theorem \ref{theorem:main_computation}.
\begin{theorem}\label{theorem:main_computation}
For $0<\epsilon<1$, we can find a $\tfrac{108}{5}\cdot (1+\epsilon)$-$\WMMS$ assignment in
 time $O(nm+\tfrac{nm\log n}{\epsilon}+nm\log m)$.
\end{theorem}
\begin{proof}
Assuming the weighted maximin share values are given in advance, the remainder of the algorithm can be implemented efficiently: each time the algorithm allocates a contiguous block of $k \geq 1$ chores to an agent in $P'$, the pointer $s$ decreases by exactly $k$. Over the entire execution, $s$ decreases by at most $m$, and every chore is allocated exactly once.

To identify such a bundle, the algorithm searches for the smallest index $\kone$ such that some agent in $P'$ has cost at most $\tfrac{18}{5}\cdot \entitlement_{\minp}$ for the interval $\{\chore_{\kone}, \chore_{\kone+1},\allowbreak \ldots, \chore_{\ktwo}\}$. This search moves $\kone$ leftwards, checking feasibility for each candidate by examining all agents in $P'$. Although $P'$ can be large, it consists only of copies of agents in $P$, which has at most $n$ distinct agents. Thus, each feasibility check takes $O(n)$ time, and it suffices to maintain the copy count per agent type in $P'$. Algorithm~\ref{algorithm:layered} represents its conceptual copies by counters: in the implementation, we store only a counter $q_i=\lfloor3\entitlement_i/\entitlement_{\minp}\rfloor$ for each agent $\agent_i\in P$. We check each agent with a positive counter once per feasibility check, decrement her counter whenever she receives a bundle, and stop considering her in the current round when the counter reaches zero. Thus, we never construct the copies explicitly.

If a bundle of size $k$ is allocated, then the total time spent to identify this bundle is $O(nk)$, as the pointer $\kone$ moves left exactly $k$ times and each move requires checking all agents in $P'$. Charging this cost to the $k$ allocated chores, we obtain an amortized cost of $O(n)$ per chore. Summing over all $m$ chores, the total time spent in identifying and allocating bundles is therefore $O(nm)$. All remaining operations in the algorithm, including updating the sets $P$, $Q$, and $D$, and computing the index $\minp'$ take $O(n)$ time per iteration and are dominated by the $O(nm)$ term.

We conclude that the overall running time of the algorithm is $O(nm)$, assuming that the weighted maximin share values are given, and that the sorted-chores assumption stated in Lemma \ref{lemma:1} holds.

To use these estimates, apply Corollary~\ref{cor:1} directly to the original entitlements. Let $\widehat\mu_i$ be the value of the returned partition for agent $\agent_i$, so that
\[
\WMMS_i\leq\widehat\mu_i\leq2(1+\epsilon)\WMMS_i.
\]
Scale her cost function by $\entitlement_i/\widehat\mu_i$. Her weighted maximin share in the scaled instance is at most $\entitlement_i$. Thus, for each agent $\agent_i$, there exists a partition $\langle B_1,\allowbreak\ldots,B_n\rangle$ with $\cost_i(B_j)\leq\entitlement_j$ for all $j$, and hence $\cost_i(\chores)\leq1$. These are the only consequences of the normalization used in Section~\ref{sec:constant}, so its analysis still applies. The sorted-chores reduction preserves this property. The algorithm assigns a bundle of scaled cost at most $\tfrac{54}{5}\cdot \entitlement_i$. Undoing the scaling gives cost at most $\tfrac{54}{5}\cdot \widehat\mu_i\leq\tfrac{108}{5}\cdot (1+\epsilon)\WMMS_i$.

All cost evaluations in the moving-knife algorithm and the\linebreak WMMS estimation procedure concern contiguous blocks in the relevant sorted order. Since costs are additive, prefix sums answer these evaluations in constant time after $O(nm)$ preprocessing.

\end{proof}

\section{Chore-oblivious Analysis}
\begin{figure*}
	\centering
	\includegraphics[width=0.55\linewidth]{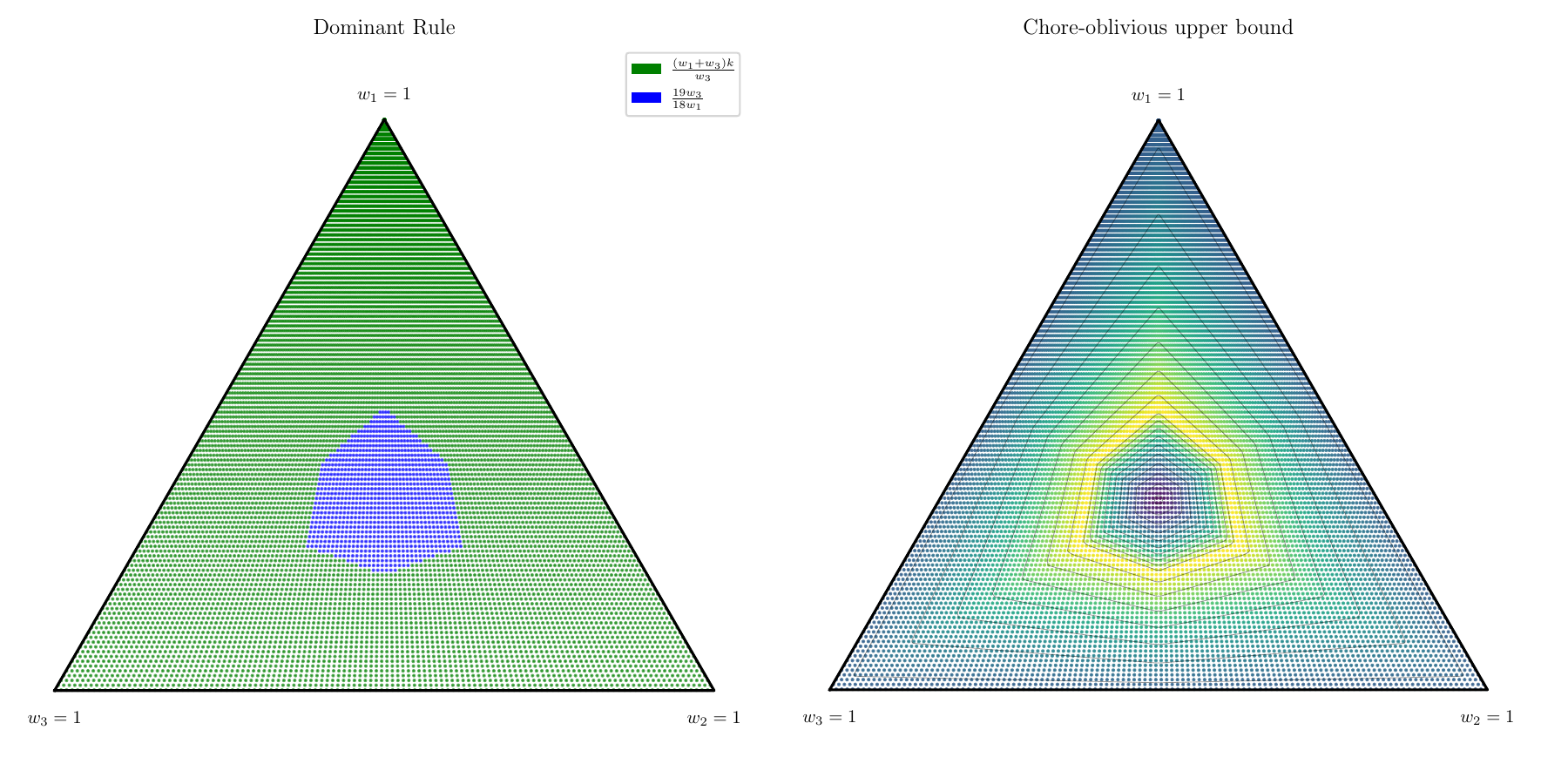}
	\caption{
Here $k=(1+\sqrt{3})/2$. The 2D simplex represents entitlement triples in barycentric coordinates: vertices correspond to full entitlement for one agent, and interior points to proportional combinations. At each point the weights are relabeled in non-decreasing order. The triangle is partitioned by the smaller of the two upper bounds on $\best(\langle \entitlement_1, \entitlement_2, \entitlement_3 \rangle)$: green for $\tfrac{(\entitlement_1 + \entitlement_3)k}{\entitlement_3}$ and blue for $\tfrac{19\entitlement_3}{18\entitlement_1}$. The right panel shows a heatmap of this upper bound, where lighter colors indicate larger values; neither panel asserts the exact value of $\best$.
	}
	\label{fig:simplexcomparisonplot}
\end{figure*}
\begin{figure*}[t]
	\centering
	\begin{minipage}{0.5\textwidth}
		\begin{tikzpicture}
			\begin{axis}[
				xlabel={$n$},
				ylabel={ Value},
				legend style={at={(1,1.2)}, anchor=east},
				grid=major,
				xtick={3,...,10},
				width=\textwidth,
				height=4cm,
				ymajorgrids=true,
				xmajorgrids=true,
				title style={font=\footnotesize}
				]
				\addplot+[mark=*, thick] table {
					n  f
					3  \vthree
					4 \vfour
					5 \vfive
					6 \vsix
					7 \vseven
					8 \veight
					9 \vnine
					10 \vten
				};
				\addlegendentry{Observed bound}
				
				\addplot+[mark=square*, thick] table {
					n  g
					3 2.584
					4 3.0
					5 3.321
					6 3.584
					7 3.807
					8 4.0
					9 4.169
					10 4.321
				};
				\addlegendentry{$\log_2 n+1$}
			\end{axis}
		\end{tikzpicture}
	\end{minipage}
	\begin{minipage}{0.49\textwidth}
		\footnotesize
		\renewcommand{\arraystretch}{1.2}
\setlength{\tabcolsep}{4pt}
\begin{tabular}{|c|c|c||c|c|c|}
	\hline
	\multicolumn{3}{|c||}{$n = 3,4,5,6$} & \multicolumn{3}{c|}{$n = 7,8,9,10$} \\
	\hline
	$n$ & \shortstack{Observed\\max.} & $\log_2 n+1$ & $n$ & \shortstack{Observed\\max.} & $\log_2 n+1$ \\
	\hline
	3 & \vthree & 2.584 & 7 & \vseven & 3.807 \\
	4 & \vfour  & 3.0   & 8 & \veight & 4.0   \\
	5 & \vfive  & 3.321 & 9 & \vnine  & 4.169 \\
	6 & \vsix   & 3.584 & 10& \vten   & 4.321 \\
	\hline
\end{tabular}
	\end{minipage}
	\caption{Largest observed chore-oblivious upper bounds over the sampled entitlement vectors, compared with the general guarantee of $\log_2 n+1$. Numerical entries are approximate.}
	\label{fig:bound_comparison}
\end{figure*}
In this section, our goal is to establish improved upper bounds on the weighted maximin share by combining  guarantees for different regimes of entitlements. To achieve this, we introduce the function \(\best(\cdot)\), which allows us to derive these bounds systematically.
 
\begin{definition}\label{def:best}
	For a vector of entitlements \(\langle\entitlement_1,  \ldots, \entitlement_n\rangle\), we define
	$
	\best(\langle \entitlement_1,  \ldots, \entitlement_n \rangle)
	$ 
	as the smallest value \(\alpha\) such that, for any instance of the asymmetric chore division problem with these entitlements, there exists an \(\alpha\)-$\WMMS$  assignment.
\end{definition}

In this section, we use simple upper bounds on \( \best(\cdot) \) along with known approximation guarantees in special cases to derive more general bounds. These known results are:

\begin{itemize}
	\item For a single agent, trivially we have:
	$
		\best(\langle 1 \rangle) = 1. \label{eq:one-agent}
	$
	
	\item For two agents with equal entitlements, the\linebreak divide-and-choose protocol guarantees that 
$
		\best(\langle \tfrac{1}{2}, \tfrac{1}{2} \rangle) = 1. \label{eq:two-equal}
$
	
	\item For two agents with arbitrary entitlements, Wang \textit{et al.} \cite{wang2024improved} show that
$
 \best(\langle \entitlement, 1 - \entitlement \rangle) \leq \tfrac{1 + \sqrt{3}}{2}. \label{eq:two-general}
$
	
	\item Feige and Norkin \cite{feige2022improved} and Huang and Segal-Halevi \cite{huang2023reduction} provided guarantees of \(\tfrac{19}{18}\)-$\MMS$ for three agents and \(\tfrac{20}{17}\)-$\MMS$ for four and five agents:
$
\best(\langle \tfrac{1}{3}, \tfrac{1}{3}, \tfrac{1}{3} \rangle) \leq \tfrac{19}{18} \label{eq:three-equal}$,
$\best(\langle \tfrac{1}{4}, \tfrac{1}{4}, \tfrac{1}{4}, \tfrac{1}{4} \rangle) \leq \tfrac{20}{17}$, and
$\best(\langle \tfrac{1}{5}, \tfrac{1}{5}, \tfrac{1}{5}, \tfrac{1}{5}, \tfrac{1}{5} \rangle) \leq \tfrac{20}{17}.
$
\end{itemize}

These results form the foundation for the upper bounds we develop for more general cases. Our goal in Lemmas~\ref{red1} and~\ref{red2} is to bound \(\best(\langle\entitlement_1,\allowbreak\ldots,\allowbreak\entitlement_n\rangle)\) by reducing instances to simpler ones. 
\begin{lemma}\label{red1}
	Let \(\entitlement_1, \entitlement_2,\allowbreak \ldots, \entitlement_n\) denote the entitlements of \(n\) agents. Suppose the agents are partitioned into \(n'\) disjoint groups, and for each group \(i \in [n']\), let \(i^* \in [n]\) be the index of a designated representative agent (not necessarily in group \(i\)), with representatives of different groups pairwise distinct. Let \(\alpha > 1\), and assume the total entitlement in group \(i\) is at most \(\alpha \entitlement_{i^*}\), i.e.,
	\begin{equation}\label{eq:group}
		\sum_{j \in \text{group } i} \entitlement_j \leq \alpha \entitlement_{i^*}.
	\end{equation}
	Then, for \(\beta = 1 / \left( \sum_{i=1}^{n'} \entitlement_{i^*} \right)\), we have
	$
	\best(\langle\entitlement_1, \entitlement_2,\allowbreak \ldots, \entitlement_n\rangle) \leq \alpha \cdot \best(\langle\beta \entitlement_{1^*}, \beta \entitlement_{2^*},\allowbreak \ldots, \beta \entitlement_{{n'}^*}\rangle).
	$
\end{lemma}
\begin{proof}
	Let
	$
	\gamma = \best(\langle\beta \entitlement_{1^*},\allowbreak \ldots, \beta \entitlement_{{n'}^*}\rangle),
	$ and consider an instance with \(n\) agents \(\agent_1,\allowbreak \ldots, \agent_n\), a set of chores \(\chores\), additive cost functions \(\cost_1,\allowbreak \ldots, \cost_n\), and entitlements \(\entitlement_1,\allowbreak \ldots, \entitlement_n\). We construct a new instance with \(n'\) super-agents. Each super-agent \(s_i\) has entitlement \(\beta \entitlement_{i^*}\), and has the cost function \(\cost_{s_i} := \cost_{i^*}\).
	The weighted maximin share for a super-agent \( s_i \), among a total of \( n' \) super-agents, denoted by \( \WMMS_{s_i} \), is defined as follows:
	\begin{equation}\label{eq:wmms-s}
		\WMMS_{s_i} = \min_{\langle A_1,\allowbreak \ldots, A_{n'} \rangle \in \Pi_{n'}(\chores)} \max_{k \in [n']} \cost_{i^*}(A_k) \cdot \frac{\entitlement_{i^*}}{\entitlement_{k^*}}.
	\end{equation}
	
	Let \(\langle B_1,\allowbreak \ldots, B_n \rangle\) be the \(\WMMS\) partition of agent \(\agent_{i^*}\) in the original instance:
	\begin{equation}\label{eq:wmms-i}
		\WMMS_{i^*} = \max_{j \in [n]} \cost_{i^*}(B_j) \cdot \frac{\entitlement_{i^*}}{\entitlement_j}.
	\end{equation}
	
	We build a partition \(\langle A_1,\allowbreak \ldots, A_{n'} \rangle\) for the reduced instance by defining each bundle \(A_k\) as the union of all \(B_j\) such that agent \(\agent_j\) belongs to group \(k\).
	By additivity of costs and using Inequality~\eqref{eq:group} we have:
	\begin{align*}
		\cost_{i^*}(A_k) &\leq \sum_{j \in \text{group } k} \cost_{i^*}(B_j) \\
		&\leq \sum_{j \in \text{group } k} \WMMS_{i^*} \cdot \frac{\entitlement_j}{\entitlement_{i^*}} &\mbox{Inequality \eqref{eq:wmms-i}}\\
		&= \frac{\WMMS_{i^*}}{\entitlement_{i^*}} \cdot \sum_{j \in \text{group } k} \entitlement_j \\
		&\leq \alpha \cdot \frac{\entitlement_{k^*}}{\entitlement_{i^*}} \cdot \WMMS_{i^*} &\mbox{Inequality \eqref{eq:group}}
	\end{align*}
	
	Multiplying both sides by \(\tfrac{\entitlement_{i^*}}{ \entitlement_{k^*}}\), we obtain
	\[
	\cost_{i^*}(A_k) \cdot \frac{\entitlement_{i^*}}{\entitlement_{k^*}} \leq \alpha \cdot \WMMS_{i^*}.
	\]
	Thus, by Equation~\eqref{eq:wmms-s} we have
$	\WMMS_{s_i} \leq \alpha \cdot \WMMS_{i^*}.
$
	Now, by the definition of \(\gamma\), there exists an assignment $\langle S_1,S_2,\allowbreak\ldots,S_{n'}\rangle$ in the reduced instance where:
	\[
	\cost_{s_i}(S_i) \leq \gamma \cdot \WMMS_{s_i} \leq \gamma \alpha \cdot \WMMS_{i^*}.
	\]
	
	We map this assignment back to the original instance by assigning each bundle \(S_i\) to the representative agent of group $i$, i.e., \(\agent_{i^*}\), and assigning empty bundles to all other agents. Then for each representative agent \(\agent_{i^*}\), we have \(\cost_{i^*}(S_i) = \cost_{s_i}(S_i)\leq \gamma \alpha \cdot \WMMS_{i^*}\), and for any non-representative agent \(\agent_j\), \(\cost_j(\emptyset) = 0 \leq \gamma \alpha \cdot \WMMS_j\).	Hence, this assignment is \(\gamma\alpha\)-\(\WMMS\). Therefore, we have	$
	\best(\langle\entitlement_1,\allowbreak \ldots, \entitlement_n\rangle) \leq  \alpha\gamma = \alpha \cdot \best(\langle\beta \entitlement_{1^*},\allowbreak \ldots, \beta \entitlement_{{n'}^*}\rangle).
	$
	\end{proof}

Lemma~\ref{red2} establishes a relationship between the approximation guarantees in the symmetric and asymmetric settings. Observe that in the symmetric case,  \(\MMS_i = \WMMS_i\) for all agents. Lemma~\ref{red2} is especially useful when the entitlements are close to each other.

\begin{lemma}\label{red2}
	Let \(\entitlement_1,\allowbreak \ldots, \entitlement_n\) be the entitlements of \(n\) agents. Then
$
	\best(\langle\entitlement_1,\allowbreak \ldots, \entitlement_n\rangle) \leq \tfrac{\max_i \entitlement_i}{\min_i \entitlement_i} \cdot \best(\langle \tfrac{1}{n}, \tfrac{1}{n},\allowbreak \ldots, \tfrac{1}{n} \rangle).
$
\end{lemma}
\begin{proof}
	We begin by considering the symmetric setting, where all agents have equal entitlements. Let \(\gamma = \best(\langle \tfrac{1}{n},\allowbreak \ldots, \tfrac{1}{n} \rangle)\), and let \(\alloc = \langle A_1,\allowbreak \ldots, A_n\rangle\) be a \(\gamma\)-\(\MMS\) assignment computed in this symmetric instance. By definition of \(\MMS\), we have for every agent \(\agent_i\):
	\begin{equation}
		\cost_i(A_i) \leq \gamma \cdot \MMS_i. \label{eq:gamma-mms}
	\end{equation}
	
	Now consider the same instance but with arbitrary entitlements \(\entitlement_1,\allowbreak \ldots, \entitlement_n\). We relate each agent's unweighted and weighted maximin share. Let \(\WMMS_i\) be the weighted maximin share of agent \(\agent_i\) in the asymmetric instance, and let \(\langle B_1,\allowbreak \ldots, B_n\rangle\) be the partition achieving it. Then:
	\begin{align}
		\WMMS_i &= \max_{j \in [n]} \left( \cost_i(B_j) \cdot \frac{\entitlement_i}{\entitlement_j} \right) \nonumber \\
		&\geq  \max_{j \in [n]}\left( \cost_i(B_j) \right) \cdot \left( \frac{\entitlement_i}{\max_j \entitlement_j} \right) \nonumber \\
		&\geq \MMS_i \cdot \left( \frac{\min_j \entitlement_j}{\max_j \entitlement_j} \right). \label{eq:wmms-lb}
	\end{align}
	
	Combining Inequalities~\eqref{eq:gamma-mms} and~\eqref{eq:wmms-lb}, we obtain:
	\[
	\cost_i(A_i) \leq \gamma \cdot \MMS_i \leq \gamma \cdot \WMMS_i \cdot \left( \frac{\max_j \entitlement_j}{\min_j \entitlement_j} \right).
	\]
	
	This shows that assignment \(\alloc\) is also \(\gamma \cdot \tfrac{\max_j \entitlement_j}{\min_j \entitlement_j}\)-\(\WMMS\)  under entitlements $\entitlement_1,\allowbreak\entitlement_2,\allowbreak\ldots,\allowbreak\entitlement_n$.
	Hence,
	\[
	\best(\langle \entitlement_1,\allowbreak \ldots, \entitlement_n \rangle) \leq \best(\langle \tfrac{1}{n},\allowbreak \ldots, \tfrac{1}{n} \rangle) \cdot \frac{\max_j \entitlement_j}{\min_j \entitlement_j}.
	\]
\end{proof}

Building on Lemmas \ref{red1} and \ref{red2} alongside the boundary cases from previous studies, Theorems~\ref{thm:3agent},~\ref{thm:4agents}, and~\ref{thm:5agent} establish upper bounds for three, four, and five agents, respectively.

\begin{theorem}\label{thm:3agent}
	For $3$ agents, there always exists a $$
	 (\tfrac{9(1+\sqrt{3}) + \sqrt{6(168 + 141\sqrt{3})}}{36})
	 \mbox{-}\WMMS$$ $(\approx 2.06447$-$\WMMS)$  assignment.
\end{theorem}

\begin{proof}
	Let \( k = \tfrac{\sqrt{3} + 1}{2} \) and define 
	$ c = \tfrac{3k + \sqrt{9k^2 + 38k}}{6}. $ 
	Observe that \( c \) is a root of the quadratic equation 
	\( 18c^2 - 18kc - 19k = 0 \). We show that for any positive entitlements \(\entitlement_1, \entitlement_2, \entitlement_3\) summing to 1, we have
$
\best(\langle \entitlement_1, \entitlement_2, \entitlement_3 \rangle) \leq c.
$
Without loss of generality, we assume \( \entitlement_1 \leq \entitlement_2 \leq \entitlement_3 \). We apply the following two upper bounds on \( \best(\langle\entitlement_1,\allowbreak\entitlement_2,\allowbreak\entitlement_3\rangle) \):
\begin{itemize}
\item Apply Lemma~\ref{red1} to partition \( \{\agent_2\}, \{\agent_1, \agent_3\} \), with \(\agent_2\) and \(\agent_3\)  as the representative agents of the first and the second group, respectively. Next, use the bound \(k\) established for the two-agent case with entitlements \(\langle\tfrac{\entitlement_2}{\entitlement_2+\entitlement_3},\tfrac{\entitlement_3}{\entitlement_2+\entitlement_3}\rangle\):
\[
	\best(\langle \entitlement_1, \entitlement_2, \entitlement_3 \rangle)
	\leq \tfrac{(\entitlement_1 + \entitlement_3)k}{\entitlement_3}.
	\label{eq:b2}
\]
	\item Lemma~\ref{red2} provides an upper bound that depends on the ratio between the largest and smallest entitlements, together with the approximation guarantee in the symmetric setting. Using the bound of \( \tfrac{19}{18} \), established by~\cite{feige2022improved} for the symmetric case, we get
$
		\best(\langle \entitlement_1, \entitlement_2, \entitlement_3 \rangle)
		\leq \tfrac{19}{18}\cdot \tfrac{\entitlement_3}{\entitlement_1}.
		\label{eq:b3}
$
\end{itemize}
 In the remainder of the proof, we show that for any valid entitlement triple \( (\entitlement_1, \entitlement_2, \entitlement_3) \), at least one of these bounds ensures an approximation ratio of at most \( c \approx 2.06447 \).
		We prove the claim by contradiction. Suppose there exist positive entitlements \( \entitlement_1, \entitlement_2, \entitlement_3 \) summing to 1 such that
	$
	\best(\langle\entitlement_1, \entitlement_2, \entitlement_3\rangle) > c.
	$
	By assumption, both bounds must be strictly greater than \( c \). That gives:
\noindent
\begin{minipage}{0.45\linewidth}
	\begin{align}
		\frac{(\entitlement_1 + \entitlement_3)k}{\entitlement_3} &> c, \label{eq:ineq1}
	\end{align}
\end{minipage}
\hfill
\begin{minipage}{0.45\linewidth}
	\begin{align}
		\frac{19\entitlement_3}{18\entitlement_1} &> c. \label{eq:ineq2}
	\end{align}
\end{minipage}
	From Inequality~\eqref{eq:ineq1}, we obtain
$
	\entitlement_1 > \entitlement_3\left( \tfrac{c}{k} - 1 \right)
$, and 
	from Inequality~\eqref{eq:ineq2}, we get 
	$
	\entitlement_1 < \tfrac{19\entitlement_3}{18c}.
	$
	Combining the two bounds we derive
	$
	\entitlement_3\left( \tfrac{c}{k} - 1 \right) < \tfrac{19\entitlement_3}{18c}.
	$
	Since \( \entitlement_3 > 0 \), we can divide both sides by $\entitlement_3$:
	$
	\tfrac{c}{k} - 1 < \tfrac{19}{18c}.
	$
	Multiplying both sides by \( 18ck \), we obtain:
	$
	18c^2 - 18ck < 19k,
$ or $
	18c^2 - 18ck - 19k < 0.
	$
	However, by definition of $c$, we have $18c^2 - 18ck - 19k =0$. Thus, at least one of Inequalities \eqref{eq:ineq1} or \eqref{eq:ineq2} must be violated, which is a contradiction.
\end{proof}

Figure~\ref{fig:simplexcomparisonplot} shows a partition of the entitlement simplex.
Next, in Theorems \ref{thm:4agents} and \ref{thm:5agent} we prove an upper bound on $\WMMS$ for four and five agents.

\begin{theorem}\label{thm:4agents}
Let $c\approx2.475855$ be the positive root of
$18c^3-18(k+1)c^2+(18k-19)c-19k=0$, where $k=(1+\sqrt3)/2$.
For four agents, there always exists a $c$-$\WMMS$ assignment.
\end{theorem}
\begin{proof}
	We show that for any positive entitlements $\entitlement_1,\allowbreak\entitlement_2,\allowbreak\entitlement_3,\allowbreak\entitlement_4$ summing to 1, 
	$
	\best(\langle \entitlement_1, \entitlement_2, \entitlement_3,\allowbreak\entitlement_4 \rangle) \leq c
	$.
	Without loss of generality, we assume $\entitlement_1\leq\entitlement_2\leq  \entitlement_3\leq \entitlement_4$, and use the following upper bounds:
	\begin{itemize}
		\item We apply Lemma \ref{red1} with partition $\{\agent_1,\agent_2,\agent_3,\agent_4\}$ with $\agent_{1^*} = \agent_4$ and then leverage $	\best(\langle 1  \rangle) = 1$ to obtain the upper bound of 
		$$
			\text{(i)}\quad  \best(\langle \entitlement_1,\allowbreak\entitlement_2,\allowbreak\entitlement_3,\allowbreak\entitlement_4 \rangle)\le \frac{1}{\entitlement_4},\\
		$$
		\item  We apply Lemma \ref{red1} using partition $\{\agent_1,\agent_2\}$ and $\{\agent_3,\agent_4\}$ and $\agent_{1^*} = \agent_3$ and $\agent_{2^*} = \agent_4$. Next, we use Lemma \ref{red2} and the upper bound of $\tfrac{\sqrt 3+1}{2}$ \cite{wang2024improved} for two agents:
\[
\text{(ii)}\quad
\best(\langle \entitlement_1,\allowbreak\entitlement_2,\allowbreak\entitlement_3,\allowbreak\entitlement_4 \rangle)
\le
\max\!\left\{
\begin{array}{c}
	\tfrac{\entitlement_4+\entitlement_3}{\entitlement_4} \\[2pt]
	\tfrac{\entitlement_1+\entitlement_2}{\entitlement_3}
\end{array}
\right\}
\min\!\left\{
\begin{array}{c}
	\tfrac{\entitlement_4}{\entitlement_3} \\[2pt]
	\tfrac{\sqrt{3}+1}{2}
\end{array}
\right\}.
\]

	\item  We apply Lemma \ref{red1} using partition $\{\agent_1,\agent_2\}$ and $\{\agent_3\}$ and $\{\agent_4\}$ and $\agent_{1^*} = \agent_2$ and $\agent_{2^*} = \agent_1$ and $\agent_{3^*} = \agent_3$. 
	Next, we use Lemma \ref{red2} and the upper bound of $\tfrac{19}{18}$ for three agents to obtain
	$$
	\text{(iii)}\quad 
	\best(\langle \entitlement_1,\allowbreak\entitlement_2,\allowbreak\entitlement_3,\allowbreak\entitlement_4  \rangle)\le \max\!\left\{\frac{\entitlement_4}{\entitlement_3},\frac{\entitlement_3}{\entitlement_1},\frac{\entitlement_1+\entitlement_2}{\entitlement_2}\right\}
	\frac{19}{18}\frac{\entitlement_3}{\entitlement_1}.
	$$
	\item Finally, we apply Lemma~\ref{red2} and use the upper bound of $\tfrac{20}{17}$ for four agents with equal entitlements~\cite{huang2023reduction}:
\begingroup
\[
\text{(iv)}\quad
\best(\langle\entitlement_1,\allowbreak\entitlement_2,\allowbreak\entitlement_3,\allowbreak\entitlement_4\rangle)
\leq\frac{20}{17}\frac{\entitlement_4}{\entitlement_1}.
\]
\endgroup
\end{itemize}
\begingroup
Now, we consider two cases according to which term in the maximum of \emph{(ii)} dominates. Suppose, toward a contradiction, that
$\best(\langle\entitlement_1,\allowbreak\entitlement_2,\allowbreak\entitlement_3,\allowbreak\entitlement_4\rangle)>c$.

\noindent\textbf{Case 1:} The second term in the maximum of \emph{(ii)} dominates, i.e.,
\[
\frac{\entitlement_1+\entitlement_2}{\entitlement_3}
>\frac{\entitlement_3+\entitlement_4}{\entitlement_4}.
\]
In this case, bound \emph{(ii)} simplifies to
\[
\best(\langle\entitlement_1,\allowbreak\entitlement_2,\allowbreak\entitlement_3,\allowbreak\entitlement_4\rangle)
\leq\frac{\entitlement_1+\entitlement_2}{\entitlement_3}
\min\!\left\{\frac{\entitlement_4}{\entitlement_3},k\right\}.
\]
To exceed $c$, the right-hand side must be strictly larger than $c$.
Since $\min\{\entitlement_4/\entitlement_3,k\}\leq k$, this is possible only if
\begin{equation}\label{eq:four-case1-ratio}
\frac{\entitlement_1+\entitlement_2}{\entitlement_3}>\frac{c}{k}.
\end{equation}
Using $\entitlement_1\leq\entitlement_2\leq\entitlement_3$, inequality~\eqref{eq:four-case1-ratio} implies
\begin{equation}\label{eq:four-ratios}
\frac{\entitlement_3}{\entitlement_1}<\frac{k}{c-k},
\qquad
\frac{\entitlement_3}{\entitlement_2}<\frac{2k}{c}.
\end{equation}
Consequently, neither the second nor the third term in the maximum of \emph{(iii)}, after multiplication by $\tfrac{19}{18}\cdot\entitlement_3/\entitlement_1$, can exceed $c$. Indeed,
\[
\frac{19}{18}\left(\frac{\entitlement_3}{\entitlement_1}\right)^2
<\frac{19}{18}\left(\frac{k}{c-k}\right)^2<2<c,
\]
and
\begin{align*}
\frac{\entitlement_1+\entitlement_2}{\entitlement_2}
\frac{19}{18}\frac{\entitlement_3}{\entitlement_1}
&=\frac{19}{18}\left(\frac{\entitlement_3}{\entitlement_2}
+\frac{\entitlement_3}{\entitlement_1}\right) &\\
&<\frac{19}{18}\left(\frac{2k}{c}+\frac{k}{c-k}\right)
&\quad\text{by }\eqref{eq:four-ratios}\\
&\approx 2.464005<c. &
\end{align*}
Hence, the first term must be active in the maximum of \emph{(iii)}, and this bound simplifies to
\[
\best(\langle\entitlement_1,\allowbreak\entitlement_2,\allowbreak\entitlement_3,\allowbreak\entitlement_4\rangle)
\leq\frac{19}{18}\frac{\entitlement_4}{\entitlement_1}.
\]
This is smaller than the upper bound in \emph{(iv)}. Together with \emph{(i)} and our assumption that $\best(\langle\entitlement_1,\allowbreak\entitlement_2,\allowbreak\entitlement_3,\allowbreak\entitlement_4\rangle)>c$, it gives
\begin{equation}\label{eq:four-extreme-weights}
\entitlement_4<\frac1c,
\qquad
\entitlement_1<\frac{19\entitlement_4}{18c}<\frac{19}{18c^2}.
\end{equation}
These inequalities are incompatible with the entitlements summing to $1$. To see this, use $\entitlement_2\leq\entitlement_3$ and~\eqref{eq:four-ratios} to obtain
\begin{align*}
1&=\entitlement_1+\entitlement_2+\entitlement_3+\entitlement_4 &\\
&\leq\entitlement_1+2\entitlement_3+\entitlement_4
&\quad\entitlement_2\leq\entitlement_3\\
&<\entitlement_4+\entitlement_1\left(1+\frac{2k}{c-k}\right)
&\quad\text{by }\eqref{eq:four-ratios}\\
&<\frac1c+\frac{19}{18c^2}+\frac{19}{9c^2(c/k-1)}
&\quad\text{by }\eqref{eq:four-extreme-weights}\\
&=1. &\quad\text{definition of }c
\end{align*}
For the final equality, subtracting $1$ from the preceding expression gives
\[
-\frac{18c^3-18(k+1)c^2+(18k-19)c-19k}{18c^2(c-k)}=0.
\]
This is a contradiction.

\noindent\textbf{Case 2:} The first term in the maximum of \emph{(ii)} dominates, i.e.,
\[
\frac{\entitlement_3+\entitlement_4}{\entitlement_4}
\geq\frac{\entitlement_1+\entitlement_2}{\entitlement_3}.
\]
In this case, bound \emph{(ii)} reduces to
\[
\best(\langle\entitlement_1,\allowbreak\entitlement_2,\allowbreak\entitlement_3,\allowbreak\entitlement_4\rangle)
\leq\frac{\entitlement_3+\entitlement_4}{\entitlement_4}
\min\!\left\{\frac{\entitlement_4}{\entitlement_3},k\right\}.
\]
If $\entitlement_4/\entitlement_3\leq k$, the right-hand side is
$1+\entitlement_4/\entitlement_3\leq1+k$.
Otherwise, it is $k(1+\entitlement_3/\entitlement_4)<k+1$.
Thus, in either case,
\[
\best(\langle\entitlement_1,\allowbreak\entitlement_2,\allowbreak\entitlement_3,\allowbreak\entitlement_4\rangle)
\leq k+1<c,
\]
again a contradiction.
\endgroup
\end{proof}

\begin{theorem}\label{thm:5agent}
For five agents with arbitrary entitlements, there always exists a
$(1+\sqrt{3})$-$\WMMS$ assignment.
\end{theorem}
\begin{proof}
Let $k=(1+\sqrt{3})/2$, so that $k>\tfrac{4}{3}$ and $2+1/k=2k$.
Without loss of generality, assume that
$0<\entitlement_1\leq\entitlement_2\leq\entitlement_3
\leq\entitlement_4\leq\entitlement_5$.
\begingroup
We apply Lemmas~\ref{red1} and~\ref{red2} together with the following upper bounds:
\begin{itemize}
\item For two agents with arbitrary entitlements, we use the bound of Wang, Li, and Lu~\cite{wang2024improved}:
\[
\best(\langle\entitlement,1-\entitlement\rangle)
\leq\frac{1+\sqrt{3}}{2}=k,
\qquad 0<\entitlement<1.
\]
\item For three agents with equal entitlements, we use the bound of Feige and Norkin~\cite{feige2022improved}:
\[
\best(\langle\tfrac13,\tfrac13,\tfrac13\rangle)\leq\frac{19}{18}.
\]
\item For five agents with equal entitlements, we use the bound of Huang and Segal-Halevi~\cite{huang2023reduction}:
\[
\best(\langle\tfrac15,\tfrac15,\tfrac15,\tfrac15,\tfrac15\rangle)
\leq\frac{20}{17}.
\]
\end{itemize}
We consider the following cases.
\endgroup

\medskip
\noindent\textbf{Case 1: $\entitlement_1+\entitlement_2\leq\entitlement_5$.}
Partition the agents into $\{\agent_1,\agent_2,\agent_5\}$ and
$\{\agent_3,\agent_4\}$, with representatives $\agent_5$ and $\agent_4$,
respectively. Since
\[
 \frac{\entitlement_1+\entitlement_2+\entitlement_5}{\entitlement_5}\leq2,
 \qquad
 \frac{\entitlement_3+\entitlement_4}{\entitlement_4}\leq2,
\]
Lemma~\ref{red1} and the two-agent bound imply
\[
 \best(\langle\entitlement_1,\allowbreak\ldots,\allowbreak\entitlement_5\rangle)\leq2k.
\]

\medskip
\noindent\textbf{Case 2: $\entitlement_1+\entitlement_2>\entitlement_5$ and
$\tfrac{\entitlement_1}{\entitlement_5}\geq\tfrac{4}{9}$.}
By Lemma~\ref{red2} and the symmetric five-agent bound,
\[
 \best(\langle\entitlement_1,\allowbreak\ldots,\allowbreak\entitlement_5\rangle)
 \leq\frac{20\entitlement_5}{17\entitlement_1}
 \leq\frac{45}{17}<\frac83<2k.
\]

\medskip
\noindent\textbf{Case 3: $\entitlement_1+\entitlement_2>\entitlement_5$ and
$\tfrac{\entitlement_1}{\entitlement_5}<\tfrac{4}{9}$.}
In this case, $\tfrac{\entitlement_2}{\entitlement_5}>\tfrac{5}{9}$.
We consider two subcases.

\smallskip
\noindent\textit{Case 3.1: $\tfrac{\entitlement_4}{\entitlement_2}\leq k$.}
Use the groups $\{\agent_5\}$, $\{\agent_1,\agent_3\}$, and
$\{\agent_2,\agent_4\}$, with representatives $\agent_2$, $\agent_3$,
and $\agent_4$, respectively. Applying Lemmas~\ref{red1} and~\ref{red2}
and the symmetric three-agent bound gives
\begin{align*}
 \best(\langle\entitlement_1,\allowbreak\ldots,\allowbreak\entitlement_5\rangle)
 &\leq\frac{19}{18}\frac{\entitlement_4}{\entitlement_2}
 \max\left\{\frac{\entitlement_5}{\entitlement_2},
 \frac{\entitlement_1+\entitlement_3}{\entitlement_3},
 \frac{\entitlement_2+\entitlement_4}{\entitlement_4}\right\}\\
 &=\frac{19}{18}\max\left\{
 \frac{\entitlement_4\entitlement_5}{\entitlement_2^2},
 \frac{\entitlement_4}{\entitlement_2}
 \left(1+\frac{\entitlement_1}{\entitlement_3}\right),
 1+\frac{\entitlement_4}{\entitlement_2}\right\}.
\end{align*}
Since $\tfrac{\entitlement_2}{\entitlement_5}>\tfrac{5}{9}$,
$\tfrac{\entitlement_1}{\entitlement_5}<\tfrac{4}{9}$, and
$\entitlement_3\geq\entitlement_2$, we have
\begin{align*}
 \frac{\entitlement_4\entitlement_5}{\entitlement_2^2}
 &<\frac{9k}{5},\\
 \frac{\entitlement_4}{\entitlement_2}
 \left(1+\frac{\entitlement_1}{\entitlement_3}\right)
 &\leq k\left(1+\frac{\entitlement_1}{\entitlement_2}\right)
 <\frac{9k}{5},\\
 1+\frac{\entitlement_4}{\entitlement_2}&\leq1+k.
\end{align*}
Consequently,
\[
 \best(\langle\entitlement_1,\allowbreak\ldots,\allowbreak\entitlement_5\rangle)
 \leq\max\left\{\frac{19k}{10},\frac{19(1+k)}{18}\right\}<2k,
\]
where the last inequality follows from $k>\tfrac{4}{3}>\tfrac{19}{17}$.

\smallskip
\noindent\textit{Case 3.2: $\tfrac{\entitlement_4}{\entitlement_2}>k$.}
Use the groups $\{\agent_1,\agent_5\}$ and
$\{\agent_2,\agent_3,\agent_4\}$, with representatives $\agent_4$
and $\agent_5$, respectively. By Lemmas~\ref{red1} and~\ref{red2}
and the symmetric two-agent bound,
\begin{align*}
 \best(\langle\entitlement_1,\allowbreak\ldots,\allowbreak\entitlement_5\rangle)
 &\leq\frac{\entitlement_5}{\entitlement_4}
 \max\left\{\frac{\entitlement_1+\entitlement_5}{\entitlement_4},
 \frac{\entitlement_2+\entitlement_3+\entitlement_4}{\entitlement_5}\right\}\\
 &=\max\left\{\frac{\entitlement_5(\entitlement_1+\entitlement_5)}{\entitlement_4^2},
 1+\frac{\entitlement_2+\entitlement_3}{\entitlement_4}\right\}.
\end{align*}
Since $\entitlement_3\leq\entitlement_4$ and
$\tfrac{\entitlement_4}{\entitlement_2}>k$,
\[
 1+\frac{\entitlement_2+\entitlement_3}{\entitlement_4}
 \leq2+\frac{\entitlement_2}{\entitlement_4}<2+\frac1k=2k.
\]
Moreover,
$\tfrac{\entitlement_1+\entitlement_5}{\entitlement_5}<\tfrac{13}{9}$ and
$\tfrac{\entitlement_4}{\entitlement_5}
 >k\cdot\tfrac{\entitlement_2}{\entitlement_5}>5k/9$, so
\[
 \frac{\entitlement_5(\entitlement_1+\entitlement_5)}{\entitlement_4^2}
 <\frac{117}{25k^2}<\frac{1053}{400}<\frac83<2k,
\]
where we used $k>\tfrac{4}{3}$. Hence,
$\best(\langle\entitlement_1,\allowbreak\ldots,\allowbreak\entitlement_5\rangle)<2k$
in this subcase as well. This completes the proof.
\end{proof}

We also run experiments to estimate the worst-case approximation factor of our chore-oblivious technique. For each $3 \leq n \leq 10$, we generate one billion samples by drawing the entitlements independently and uniformly, and then normalizing the resulting vector so that its coordinates sum to $1$. For each sampled vector, Lemmas~\ref{red1} and~\ref{red2} give an upper bound on $\best(\langle \entitlement_1, \entitlement_2,\allowbreak \ldots, \entitlement_n \rangle)$. The largest value observed is a lower bound on the worst-case value of the upper-bound function evaluated in the experiments, not a lower bound on $\best$. We report these empirical lower bounds in Figure~\ref{fig:bound_comparison}. For $n=3,4,5$, the largest observed values are close to our analytical upper bounds on $\best$. To establish near-tightness of the evaluated upper-bound function, one would also need to show that the analytical bounds apply to that same function; the numerical comparison alone does not establish this. For $6 \leq n \leq 10$, the reported values suggest that the corresponding worst-case upper bounds may also be close to the observed values, but this remains an empirical observation and does not establish a guarantee for unsampled entitlement vectors.

\section{Conclusion}
\begingroup
We established a $\tfrac{54}{5}$-$\WMMS$ guarantee for indivisible chores with unequal entitlements using a combinatorial layered moving-knife algorithm. Approximating the weighted maximin shares gives a $(\tfrac{54}{5}+\epsilon)$ guarantee in polynomial time, or a $\tfrac{108}{5}\cdot(1+\epsilon)$ guarantee in almost linear time. Our chore-oblivious analysis also gives smaller approximation factors for three, four, and five agents.

Several questions remain open. One direction is to improve the upper bound for an arbitrary number of agents. Another is to obtain an almost-linear-time implementation with a guarantee closer to $\tfrac{54}{5}$. For small numbers of agents, it would be useful to derive analytical guarantees for $n\geq6$ that approach the observed values, and to determine the best bounds obtainable within a precisely specified family of chore-oblivious reductions.
\endgroup

\section*{Use of AI}
The mathematical results in the earlier version were obtained without AI assistance. That version used a reduction that rounded entitlements down to powers of two to simplify the proof, at the expense of a factor of two in the approximation guarantee. In this revision, OpenAI Codex assisted in tuning the parameters and developing the proof without this reduction, yielding the improved factor of $\tfrac{54}{5}$. It also assisted with writing and with extending the existing chore-oblivious analysis for three and four agents to five agents, including the choice of thresholds in the case analysis. The authors are responsible for all results and claims.

\label{end:main-content}
\clearpage
\bibliographystyle{plainnat}
\bibliography{draft}
\end{document}